\documentclass[11pt]{article}
\usepackage{caption}
\usepackage[T1]{fontenc}
\usepackage[margin=1in]{geometry}
\usepackage{tikz-cd}
\usepackage{mathtools}
\usepackage{color}
\usepackage{proof}
\usepackage{caption}
\usepackage{bbm}
\usepackage{amssymb}
\usepackage{microtype}
\usepackage{xcolor}
\usepackage{enumitem}
\usepackage{amsthm}
\usepackage{thmtools,thm-restate}
\usepackage{graphicx}
\usepackage{algorithmicx}
\usepackage{algpseudocode}
\usepackage{qcircuit}
\usepackage{braket}
\usepackage{float}
\usepackage{subcaption}
\usepackage{multirow}

\usepackage{mathrsfs}

\usepackage{soul}

\usepackage{amsmath}
\usepackage{amsfonts}
\usepackage{url}
\DeclareMathAlphabet{\mathpzc}{OT1}{pzc}{m}{it}

\newcommand{\ketbra}[2]{\lvert #1 \rangle \! \langle #2 \rvert}
\usepackage[colorlinks,citecolor=blue]{hyperref}
\newtheorem{theorem}{Theorem}[section]
\newtheorem{lemma}[theorem]{Lemma}
\newtheorem{definition}[theorem]{Definition}

\newtheorem{corollary}[theorem]{Corollary}

\newcommand{\eq}[1]{\hyperref[eq:#1]{(\ref*{eq:#1})}}
\renewcommand{\sec}[1]{\hyperref[sec:#1]{Section~\ref*{sec:#1}}}
\newcommand{\thm}[1]{\hyperref[thm:#1]{Theorem~\ref*{thm:#1}}}
\newcommand{\lem}[1]{\hyperref[lem:#1]{Lemma~\ref*{lem:#1}}}
\newcommand{\cor}[1]{\hyperref[cor:#1]{Corollary~\ref*{cor:#1}}}
\newcommand{\itm}[1]{\hyperref[itm:#1]{\ref*{itm:#1}}}
\newcommand{\app}[1]{\hyperref[app:#1]{Appendix~\ref*{app:#1}}}
\newcommand{\dfn}[1]{\hyperref[dfn:#1]{Definition~\ref*{dfn:#1}}}
\newcommand{\fig}[1]{\hyperref[fig:#1]{Figure~\ref*{fig:#1}}}
\newcommand{\clm}[1]{\hyperref[clm:#1]{Claim~\ref*{clm:#1}}}
\newcommand{\alg}[1]{\hyperref[alg:#1]{Algorithm~\ref*{alg:#1}}}
\newcommand{\stp}[1]{\hyperref[stp:#1]{Step~\ref*{stp:#1}}}
\newcommand{\asm}[1]{\hyperref[asm:#1]{Assumption~\ref*{asm:#1}}}
\newcommand{\prot}[1]{\hyperref[prot:#1]{Protocol~\ref*{prot:#1}}}
\newcommand{\prob}[1]{\hyperref[prob:#1]{Problem~\ref*{prob:#1}}}
\newcommand{\rmk}[1]{\hyperref[rmk:#1]{Remark~\ref*{rmk:#1}}}
\newcommand{\cons}[1]{\hyperref[cons:#1]{Construction~\ref*{cons:#1}}}
\newcommand{\conj}[1]{\hyperref[conj:#1]{Conjecture~\ref*{conj:#1}}}
\newcommand{\tbl}[1]{\hyperref[tbl:#1]{Table~\ref*{tbl:#1}}}
\usepackage{tabularx}

\newcommand{\tr}[0]{\mathrm{tr}}

\let\originalleft\left
\let\originalright\right
\renewcommand{\left}{\mathopen{}\mathclose\bgroup\originalleft}
\renewcommand{\right}{\aftergroup\egroup\originalright}

\newcommand{\A}[0]{\mathcal{A}}

\newcommand{\D}[0]{\mathcal{D}}

\newcommand{\N}[0]{\mathcal{N}}
\renewcommand{\O}[0]{\mathcal{O}}

\DeclareMathOperator*{\Exp}{\mathbb{E}}

\DeclareMathOperator{\poly}{poly}
\DeclareMathOperator{\polylog}{polylog}
\DeclareMathOperator{\trnc}{trnc}

\newcommand{\class}[1]{\mathsf{#1}}
\renewcommand{\P}[0]{\class{P}}
\newcommand{\sharpP}[0]{{\#\class{P}}}
\newcommand{\BPP}[0]{\class{BPP}}

\newcommand{\BQP}[0]{\class{BQP}}
\newcommand{\NP}[0]{\class{NP}}
\newcommand{\PSPACE}[0]{\class{PSPACE}}

\newcommand{\MA}[0]{\class{MA}}
\newcommand{\PH}[0]{\class{PH}}

\newcommand{\NISQ}[0]{\class{NISQ}}

\newcommand{\RE}[0]{\class{RE}}

\newcommand{\MIP}[0]{\class{MIP}}

\newcommand{\Id}[0]{\mathbbm{1}}

\newcommand{\bit}{\{0,1\}}

 \newcommand{\nai}[1]{} 
\usepackage{authblk}

\usepackage{qcircuit}

\begin{document}

\title{Oracle Separation between Noisy Quantum Polynomial Time and the Polynomial Hierarchy}

\author[1]{ Nai-Hui Chia}
\author[2]{ Min-Hsiu Hsieh}
\author[3]{ Shih-Han~Hung}
\author[2,4]{En-Jui Kuo}

\affil[1]{\footnotesize Department of Computer Science, Rice University}
\affil[2]{\footnotesize Foxconn Research, Taipei, Taiwan}
\affil[3]{\footnotesize Department of Electrical Engineering, National Taiwan University}
\affil[4]{\footnotesize Physics Division, National Center for Theoretical Sciences, Taipei 106319, Taiwan }

\date{}
\maketitle

\begin{abstract}
This work investigates the oracle separation between the physically motivated complexity class of noisy quantum circuits, drawing inspiration from definitions like those in \cite{chen2022complexity}. We establish that with a constant error rate, separation can be achieved in terms of $\NP$. When the error rate is $\Omega(\log n/n)$, we can extend this result to the separation of $\PH$. Notably, our oracles, in all separations, do not necessitate error correction schemes or fault tolerance, as all quantum circuits are of constant depth. This indicates that even quantum computers with minor errors, without error correction, may surpass classical complexity classes under various scenarios and assumptions. We also explore various common noise settings and present new classical hardness results, generalizing \cite{raz2022oracle, bassirian2021certified}, which are of independent interest.
\end{abstract}
\section{Introduction}

\nai{The first part of the introduction is redundant. The introduction shall start with the problem we are going to tackle, which is to understand the power of NISQ devices. I will revise the intro later.}

One of the most important goals in the field of quantum complexity theory is to characterize the computational complexity of the complexity class $\BQP$ \cite{bernstein1993quantum,nielsen2002quantum}, which comprises languages that can be efficiently solved by a quantum computer. Researchers have studied the computational power of $\BQP$ in comparison to classical complexity classes \cite{bernstein1993quantum, adleman1997quantum, fortnow1999complexity}, leading to a chain of inclusions:

\begin{equation}
{\mathsf{P\subseteq BPP\subseteq BQP\subseteq PP\subseteq \P^{\sharpP}\subseteq PSPACE\subseteq EXP}}.
\end{equation}

The inclusion of $\BQP$ in this chain highlights the potential computational advantage of quantum computers over classical computers. 
However, we cannot hope to show $\BPP\neq\BQP$ without proving that $\P\neq\PSPACE$. 
Instead, researchers seek to design problems that are difficult for specific complexity classes to solve but can be solved by other classes relative to an oracle, e.g., oracle separation between $\BQP$ and classical complexity classes. 
These oracle separations provide insights into the relative power between different complexity classes and help us understand the limitations of quantum computers in solving certain problems.
 
Separating $\BQP$ from classical complexity classes relative to an oracle has a long history. 
The Deutsch-Jozsa algorithm provides a separation between $\P$ and $\BQP$ \cite{berthiaume1992quantum} and has been further extended to $\NP$ and $\BQP$ \cite{berthiaume1994oracle}. 
Simon's problem provides a separation between $\BQP$ and $\BPP$ \cite{simon1997power}. 
Later, Bernstein and Vazirani defined a problem called Recursive Fourier Sampling \cite{bernstein1993quantum}, which separates $\BQP$ and $\MA$.
Subsequently, Watrous \cite{watrous2000succinct} gave an oracle $\O$ relative to which $\BQP\not\subset\MA$. 
In 2009, Aaronson defined a problem called Forrelation \cite{aaronson2010bqp} and showed that if the Generalized Linial-Nisan Conjecture were true, then there would exist an oracle that separates $\BQP$ and the polynomial hierarchy. However, the conjecture turned out to be false \cite{aaronson2011counterexample}.
In 2018, Raz and Tal \cite{raz2022oracle} modified the original Forrelation problem into a different form
and proved an oracle separation between $\BQP$ and $\PH$. The proof can be simplified by stochastic calculus \cite{wu2020stochastic}.
On the negative side, Bennett, Bernstein, Brassard, and Vazirani \cite{bennett1997strengths} showed that there exists an oracle relative to which $\NP \not\subset \BQP$. 
In \cite{aaronson2021acrobatics}, Aaronson, Ingram, and Kretschmer have proved numerous $\BQP$-type oracle separations based on \cite{raz2022oracle}. For example, there exists an oracle relative to which $\BQP = \P^{\sharpP}$ and yet $\PH$ is infinite.

$\BQP$ refers to the ideal quantum computation without being subject to any noise.  However, in the real world, physical devices are prone to noise during the initialization of qubits, gates, and measurements, which can prevent longer quantum computations \cite{preskill2018quantum}.
Luckily, fault-tolerant quantum computing and quantum error correction \cite{devitt2013quantum, lidar2013quantum, aharonov1997fault, aharonov1996limitations, raussendorf2007fault, bacon2000universal, shor1996fault} allow to mitigate these errors. 
While fault-tolerant quantum computing has been shown to solve hard computational tasks such as factoring a large number, there are already some devices that seem to outperform classical computation but may not be able to perform universal quantum computation. 
Most of them include sampling-based quantum supremacy experiments such as BosonSampling \cite{aaronson2011computational, zhong2020quantum} and random circuit sampling \cite{arute2019quantum}. 
These devices are often referred to as NISQ devices \cite{preskill2018quantum,bharti2022noisy}, which stands for "Noisy Intermediate-Scale Quantum". They are a type of quantum computer that is currently available, has a limited number of qubits, and experiences a high level of noise and errors in its operations. These types of quantum computers are considered to be intermediate in terms of their capabilities and are not yet able to perform large-scale, error-corrected quantum computations. However, NISQ devices can still be useful for certain types of quantum algorithms including Variational quantum eigensolver (VQE) \cite{kandala2017hardware}, Quantum optimization algorithms (QAOA) \cite{streif2019comparison} 
and generating cryptographically certified random bits \cite{aaronson2023certified}, as well as for exploring the potential of quantum computing. 
Standard NISQ platforms include superconducting qubits \cite{wallraff2004strong, krantz2019quantum, larsen2015semiconductor}, trapped ions \cite{bruzewicz2019trapped}, optical systems \cite{o2007optical}, and Rydberg atoms \cite{wu2021concise}. These platforms represent the most promising technologies for building practical quantum computers with a few dozen qubits, paving the way toward quantum advantage. 

It is intriguing to consider whether similar quantum supremacy can be achieved by noisy quantum circuits. If it can be shown that $\text{Post-(noisy quantum circuit)}=\text{PostBQP}$ and an exact classical simulation exists, then the Polynomial Hierarchy would collapse for a reason akin to that in IQP \cite{lund2017quantum}. Therefore, demonstrating such would provide substantial evidence against the possibility of classical simulation of noisy quantum circuits.

To incorporate such a concept into complexity, \cite{chen2022complexity} defines a computational complexity class called $\NISQ$ by modifying the original $\BQP$ by adding the local depolarization noise after each gate, preparing all zero states and performing error measurement. 
Specifically, an informal yet informative definition of $\BQP_{\lambda}$ could be stated as follows, and its rigorous definition will be provided in \app{nisq}. We will modify their model by writing error function using $\lambda$ explicitly.
\begin{definition}[$\BQP_{\lambda}$ complexity class, informal] $\BQP_{\lambda}$ contains all problems that can be solved
by a polynomial-time probabilistic classical algorithm with access to a noisy quantum device. To
solve a problem of size $n$, the classical algorithm can access a noisy quantum device that can:
\begin{enumerate}
    \item Prepare a noisy $\poly(n)$-qubit all-zero state;
    \item Apply arbitrarily many layers of noisy two-qubit gates;
    \item Perform noisy computational basis measurements on all the qubits simultaneously.
\end{enumerate}  
\end{definition} 

\noindent Here, noisy quantum devices mean that each qubit in each time step is subject to the depolarizing noise with a function $\lambda(n):\mathbb{N} \to \mathbb{R}$ which $1>\lambda(n)> 0$.
Note that the implementation of fault-tolerant quantum computation schemes in such a model is forbidden. \nai{We shall discuss why fault-tolerant quantum computing is forbidden.}
One might think that if $\lambda(n)$ is sufficiently small, we could naturally implement error correction. However, this is prohibited as intermediate measurements are not allowed. The contribution of this paper is based entirely on using the original quantum circuit for the noisy one and then addressing the problem by directly dealing with the noise.

Based on the previous discussion, the following natural question arises: what is the power of such a physically inspired complexity class $\BQP_{\lambda(n)}$ for given $\lambda$? 
Chen, Cotler, Huang, and Li \cite{chen2022complexity} provided a first non-trivial answer from the complexity theory point of view. In particular, they showed that $\NISQ=\cup_{\lambda(n)=\text{constant}>0}\BQP_{\lambda(n)}$ can provide exponential speedup in particular problems such as Simon's problem \cite{simon1997power} and the Bernstein-Vazirani problem \cite{bernstein1993quantum}. 
In terms of oracle separation, they provide an oracle separation between $\NISQ$ and $\BPP$ using a robustified version of Simon's oracle. 
However, this modification is highly nontrivial, as it makes the original Simon's oracle robust to errors without modifying the oracle, or it requires an additional conjecture \cite[Appendix~I]{chen2022complexity}. 
Such modification enhances the original oracle and makes the separation unsurprising since it doesn't make the oracle more useful to classical computers, but a $\NISQ$ machine can enjoy the robustness.

Initially, one may think that when $\lambda$ is sufficiently small, a $\BQP_{\lambda}$ machine can exhibit the full computational power of a noiseless quantum computer. 
However, in relativized worlds, as shown in \cite{chen2022complexity}, this is not the case. The first example is that $\NISQ$ cannot achieve a quadratic speedup on unstructured search. The second example is that $\NISQ$ requires $\Omega((1-\lambda)^{-n})$ copies of $\rho$ to learn $|\tr(P\rho)|$ for all $P \in \{I, X, Y, Z\}^{\otimes n}$ up to a constant error, whereas only $O(n)$ copies are needed for $\BQP$. 
In summary, $\NISQ$ can solve Simon's and Bernstein-Vazirani's problems, which are hard for $\BPP$. But $\NISQ$ can not achieve a quadratic speedup on unstructured search using standard Grover's algorithm. 
This indicates that $\NISQ$ is nontrivially positioned between $\BPP$ and $\BQP$. Namely, we expect that $\BPP \subsetneq \NISQ \subsetneq \BQP$. Therefore, we believe that $\NISQ$ deserves further investigation.

\subsection{Our contributions}
In our work, we intend to address the following question:  
\vspace{-2mm}
\begin{quote}
{\em For a fixed function $\lambda$, can $\BQP_{\lambda}$ solve problems (and find the corresponding oracles that provide the separation) that are hard for classical complexity classes? 
}
\end{quote}

We will show that $\BQP_{\lambda}$ is still very powerful by providing evidence to support this claim. This illustrates how this physically inspired model can still achieve oracle separation and outperform classical complexity classes.

\textbf{Remark:} As described in the definition of $\BQP_{\lambda}$, intermediate measurements are not allowed. Consequently, our main contributions focus on directly addressing the noise.
\begin{theorem}[Informal]\label{thm:sh-1}
There exist a constant $\lambda>0$ and an oracle relative to which $\BQP_{\lambda} \not\subset \NP$.
\end{theorem}

This theorem has two implications. First, as shown in \cite{chen2022complexity}, one can use a majority vote to solve the Bernstein-Vazirani problem, in which each bit of the hidden string is either $0$ or $1$ and there is only one Fourier component. But this theorem shows that even with exponential Fourier components, but with a promise of $\frac{1}{\poly(n)}$ gaps,  $\BQP_{\lambda(n)=\lambda_0}$ can still distinguish them, even with a constant error rate. This also shows that $\NISQ$ \cite{chen2022complexity} can also enjoy such separation.

This separation is proven by leveraging Deutsch–Jozsa algorithm, where we are given an oracle that computes a function $f\colon \{0,1\}^{n}\to \{0,1\}.$ We are promised that the function is either constant or balanced. Standard $\BQP$ only needs a single query while $\NP$ requires exponential queries \cite{berthiaume1994oracle}. We show that there exists $\lambda_0>0$ constant which that $\BQP_{\lambda(n)=\lambda_0}$ can solve the Deutsch–Jozsa problem using $O(n^2 \log n)$ queries. Hence the theorem follows.

\begin{theorem}[Informal]\label{thm:sh-3} There exist $\lambda: \mathbb{N} \to \mathbb{R} $, where $\lambda(n)=\Omega(\log n/n)$, and an oracle relative to which $\BQP_{\lambda} \not\subset \PH$. 
\end{theorem}

This theorem indicates that quantum computers, even with small errors and without any fault-tolerant scheme, can still outperform the entire polynomial hierarchy, which is pretty remarkable. 
The theorem shows that the entire computational power of quantum computation is really based on the unitary transformation of probability amplitudes and superposition of the initial state and again, with only constant depth. 
From the experimental side, this indicates for particular error regions like $\log n /n$ that even only with $\BQP_{\lambda}$ without fault tolerance can still demonstrate quantum supremacy consisting of a problem that can not be solved in polynomial Hierarchy.

Here we require two limitations, one is that our method can only show such separation in $\Omega(\log n/n)$ and also require the error only happened before the oracle. However, the circuit only is just simply fourier sampling so it only uses oracle only once.

Finally, by leveraging Raz's theorem \cite{raz2022oracle}, Aaronson, DeVon, William can prove more oracle separations \cite{aaronson2021acrobatics}. We can also follow the proof and show an oracle relative to which $\BQP_{\lambda} \not\subset \NP/\mathsf{poly}$ exists.
Here we recall $\lambda(n)= \Omega(\frac{\log n}{n}).$
\begin{theorem}\label{thm:acr}
There exists an oracle relative to which $\BQP_{\lambda} = \P^\sharpP$ and $\PH$ is infinite and there exists an oracle relative to which $\BQP_{\lambda} \not\subset \NP/\mathsf{poly}$.    
\end{theorem}
These theorems as \cite{aaronson2021acrobatics} indicated that the randomness is indeed inherent in the quantum computation compared to the standard classical randomized algorithm $\BPP.$ Here we show that even within a small noise region, such randomness can still be preserved and can not be fixed as the standard $\BPP$.

The first part of the proof is straightforward and follows the steps of Theorem 29 in \cite{aaronson2021acrobatics} closely, substituting the complexity class $\BQP$ with the $\BQP_{\lambda}$ model. The second part mirrors Corollary 31 in \cite{aaronson2021acrobatics}  similarly, by replacing $\BQP$ with $\BQP_{\lambda}$. Consequently, we omit the detailed proof.

\subsection{Related work}
Here we briefly discuss some related works. Two particular families of noisy models are considered. First, the noise originates solely from the oracle; this is referred to as the Faulty Oracle Models. Second, the noise comes from the quantum gates \cite{chen2022complexity, hamoudi2024nisq} as well as from our own work. The key difference between the faulty oracle models and the noisy quantum model is that the former assumes noise occurs only inside the oracle, whereas in the noisy quantum model, quantum computation outside of the oracles is also subject to noise from external sources. Moreover, most studies on faulty oracles consider global noise affecting the entire oracle, while the quantum circuit typically suffers from local, qubit-wise noise.   

\textbf{The Faulty Oracle Models:} Several studies \cite{muthukrishnan2019sensitivity, shenvi2003effects, long2000dominant, regev2008impossibility, salas2008noise} have examined different faulty oracles and demonstrated various lower bounds on the performance of quantum algorithms using such oracles, primarily focusing on Grover's search algorithm. \cite{muthukrishnan2019sensitivity} shows that the glued-trees problem demonstrates an exponential quantum speedup only when using quantum annealing with excited state evolution in an oracular setting, making it unique to date. \cite{shenvi2003effects, long2000dominant} examine the robustness of Grover's search algorithm to a random phase error in the oracle and analyze its complexity. \cite{regev2008impossibility} further shows that Grover's unstructured search problem has no quantum speed-up in the setting where each oracle call has some small probability of failing. \cite{salas2008noise} discuss the complexity of Grover's algorithm on various error regions.

\textbf{NISQ models}
Here we discuss some related works on the complexity of noisy quantum circuits. \cite{chen2022complexity} define and study the complexity class $\NISQ$ and give evidence that $\BPP \not \subset \NISQ \not \subset \BQP$. However, in their proof of separation $\BPP \not \subset \NISQ$ by using Simon's algorithm needs to modify the oracle to do the error correction to decrease the error from constant to $\frac{1}{n}.$ 
The basic difference between our model, $\BQP_{\lambda(n)}$, and $\NISQ$ \cite{chen2022complexity} is their focus on the constant error rate. In contrast, we allow the error rate to be a function of the number of qubits, $n$. Formally, $\NISQ = \bigcup_{\lambda(n) = \text{constant}>0} \BQP_{\lambda(n)}$.
\cite{aharonov1996limitations}  introduces the concept of noisy reversible gates, which are reversible gates that can be implemented in the presence of noise. They then prove that there exist functions that cannot be implemented by a polynomial number of noisy reversible gates. This result implies that noisy reversible computation is strictly less powerful than classical computation. \cite{hamoudi2024nisq} considers three models of noisy quantum circuits, including those with constant depth, limited quantum queries, or gates subject to dephasing or depolarizing channels.

In our work, we aim to demonstrate positive outcomes from noisy quantum circuits. On the complementary side, some negative results suggest that the advantages of quantum circuits can be diminished by noise.
\paragraph{Negative results.}
We include a few notable ones. $\MIP^{*}$ Vanishes in the Presence of Noise' \cite{dong2023computational} is compared to the well-known result $\MIP^{*}=\RE$ \cite{ji2021mip}. Additionally, we reference the classical algorithm for noisy random quantum sampling \cite{aharonov2023polynomial}. There is a new paper about NISQ query complexity for collision finding \cite{hamoudi2024nisq}.



Based on the definition of a noisy quantum circuit, we can demonstrate the existence of an oracle $\O$ such that a noisy quantum circuit augmented by $\O$ does not fall under the subsets of $A^{\O}$, where $A=\NP, \PH$, under certain additional assumptions. 
One immediate question is that whether we can enhance our $\PH$ separation for $\lambda=O(1).$ We will provide the reason why we think this is impossible without using quantum error correction or robustified oracle. 
 

\section{Preliminaries}\label{sec:pre}
In this section, we briefly review two well-known standard quantum algorithms: the Deutsch-Jozsa algorithm and the Forrelation algorithm. These algorithms will play important roles in our paper. 

Throughout this paper, we assume that the input size is denoted as $N:=2^n$ and that readers are familiar with basic complexity classes, including $\P, \PH, \NP, \BQP$, and $\NP/\poly$\footnote{See the \href{https://complexityzoo.net/Complexity_Zoo}{Complexity Zoo} for definitions.}. We define $\mathcal{N}_{\mathbb{R}}(0,1)$ to be the normal distribution where the mean is $0$ and the standard deviation is $1.$ We denote $H_N$ or $H$ to represent the standard Hadamard gate of $n$ qubits.

\subsection{Deutsch-Jozsa algorithm}\label{sec:DJ}

In this section, we review the oracle separation between $\BQP$ and $\NP$ by introducing the Deutsch-Jozsa algorithm \cite{berthiaume1992quantum}. We are given an oracle that implements a function: ${\displaystyle f\colon \{0,1\}^{n}\rightarrow \{0,1\}}.$ We are promised that the function is either constant, meaning $f(x)=1$ or $f(x)=0$ for all $x$, or balanced. The task is to determine if $f$ is constant or balanced using the oracle.
The quantum algorithm for this problem requires only a single query, as shown in the circuit in Figure
\ref{Fig:circuitDJ}.
The probability of measuring $\ket{0}^{n}$ is $1$ if $f(x)$ is constant and $0$ if $f(x)$ is balanced.

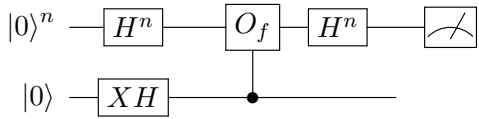
\begin{figure}[!htbp]
\begin{center}
\begin{minipage}{10cm}
\Qcircuit @C=1em @R=1em {
\lstick{\ket{0}^n} & \gate{H^n}  & \qw        & \gate{O_f}      & \gate{H^n}      & \qw & \meter\\
\lstick{\ket{0}} & \gate{XH}      & \qw        & \ctrl{-1}          & \qw      & \qw 
}
\end{minipage}
\end{center}
\caption[Quantum Circuit for Deutsch–Jozsa algorithm]{{Quantum circuit for the Deutsch–Jozsa algorithm.}
}
\label{Fig:circuitDJ}
\end{figure}

On the other hand, Berthiaume and Brassard constructed the following oracle separation using the Deutsch–Jozsa problem:
\begin{theorem}\label{thm:bb}
There exists an oracle $\mathcal{O}$ relative to which $\BQP \not\subset \NP$. More concretely, there is a set that can be recognized in worst-case linear time by a quantum computer, yet any nondeterministic Turing machine that accepts it must take exponential time on infinitely many inputs.
This task can be solved using the Deutsch–Jozsa algorithm. \cite{berthiaume1994oracle, berthiaume1992quantum}
\end{theorem}

Berthiaume and Brassard \cite{berthiaume1994oracle, berthiaume1992quantum} construct an $\mathcal{O}$ such that a nondeterministic Turing machine would need exponential time to determine $S(\mathcal{O})$ but only takes linear time to determine $S(\mathcal{O})$ using a quantum computer.

We now present a brief proof sketch. Given any oracle $\mathcal{O}$, it can be considered as a subset of $\Sigma = \{0,1\}^*$. We state that $B(\mathcal{O})$ holds if $|\mathcal{O} \cap \{0,1\}^n| = 0$ or $|\mathcal{O} \cap \{0,1\}^n| = 2^{n-1}$. We then define $S(\mathcal{O}) = \{1^n \mid \mathcal{O} \cap \{0,1\}^n = \emptyset\}$. One can determine $S(\mathcal{O})$ by using the Deutsch--Jozsa algorithm to construct a function $f: \{0, 1\}^n \to \{0,1\}$ such that $f(o) = 0$ if $o \notin \mathcal{O}$, and $f(o) = 1$ otherwise. Determining whether $f$ is a constant or balanced function is equivalent to determining whether $1^n \in S(\mathcal{O})$, which shows that such a problem belongs to $\BQP^{\mathcal{O}}$. On the other hand, \cite{berthiaume1994oracle, berthiaume1992quantum} utilize the diagonalization approach to construct a sparse $\mathcal{O}$ such that it takes exponential time to determine $S(\mathcal{O})$. We leave the details to the original papers.

\subsection{Forrelation algorithm}\label{sec:infor}

To construct an oracle separation between $\BQP$ and $\PH$, we introduce the Forrelation problem, which has various versions. The first version was proposed by Aaronson \cite{aaronson2010bqp}. It involves two Boolean functions $f, g: \{0, 1\}^n \rightarrow \{-1, 1\}$ and quantifies their Forrelation $\Phi_{f,g}$ as:
\begin{equation}
\Phi_{f,g} = \frac{1}{2^{\frac{3n}{2}}} \sum_{x, y \in \{0, 1\}^n} f(x)(-1)^{x \cdot y} g(y).
\end{equation}
The challenge of this problem \cite{aaronson2010bqp} lies in deciding whether $|\Phi_{f,g}|$ is negligible (e.g., $\leq \frac{1}{100}$) or significant (e.g., $\geq \frac{3}{5}$), under the promise that one of these scenarios holds true.

Quantum circuits provide an efficient means to solve this problem. One approach uses a control qubit prepared in the state $\ket{+}$ and applies sequences of operations conditioned on the control qubit's state, ultimately measuring it in the $\{\ket{+}, \ket{-}\}$ basis. The acceptance probability of this circuit directly relates to the Forrelation value:
\begin{equation}
    \text{Prob(accept)} = \frac{1+ \Phi_{f,g}}{2}.
\end{equation}

Aaronson and Ambainis \cite{aaronson2015forrelation} further demonstrate that classical randomized algorithms face a lower bound of $\Omega\left(\frac{N}{\log N}\right)$ queries for solving Forrelation, underscoring the quantum advantage in this context. A two-query quantum circuit variant also exists, further illustrating the flexibility in designing quantum solutions to efficiently evaluate Forrelation. We now give the oracle separation between $\BQP$ and $\BPP$ by introducing the \textbf{Real Forrelation problem} \cite{aaronson2015forrelation}.
\begin{definition}[Real Forrelation problem]
    In the Real Forrelation problem, we have oracle access to two real functions $f, g : \{0, 1\}^n \to \mathbb{R}$ and must determine between two scenarios:
    \begin{itemize}
        \item[(i)] Every $f(x)$ and $g(y)$ value is an independent $\mathcal{N}_{\mathbb{R}}(0, 1)$ Gaussian.
        \item[(ii)] Every $f(x)$ value is an independent $\mathcal{N}_{\mathbb{R}}(0, 1)$ Gaussian, and every $g(y)$ value is defined as:
        \begin{equation*}
            \hat{f}(y) = \frac{1}{\sqrt{2^n}} \sum_{x \in \{0, 1\}^n} (-1)^{x \cdot y} f(x).
        \end{equation*}
    \end{itemize}
\end{definition}

Aaronson proved that the Real Forrelation problem belongs to $\BQP$ when accessing such an oracle. The following theorem establishes the lower bound for the Real Forrelation problem for classical randomized algorithms:
\begin{theorem}[Theorem 2 in \cite{aaronson2015forrelation}]\label{thm:aa-ga}
Any randomized algorithm solving the Real Forrelation requires $\Omega\left(\frac{N}{\log N}\right)$ queries. In contrast, the corresponding quantum algorithm only needs $O(1)$ queries. Hence, an oracle $\mathcal{O}$ can be constructed using Real Forrelation such that $\BQP^{\mathcal{O}} \not\subset \BPP^{\mathcal{O}}$.
\end{theorem}

Raz and Tal \cite{raz2022oracle} then modified the distribution proposed by Aaronson and Ambainis \cite{aaronson2015forrelation} into a different distribution and proved an oracle separation between $\BQP$ and $\PH$. Let $\mathcal{D}_1$ and $\mathcal{D}_2$ be two probability distributions over a finite set $X$. An algorithm $\mathcal{A}$ is said to distinguish between $\mathcal{D}_1$ and $\mathcal{D}_2$ with an advantage $\epsilon$ if
\begin{equation}
    \epsilon = \left|\Pr_{x \sim \mathcal{D}_1} [\mathcal{A} \text{ accepts } x] - \Pr_{x' \sim \mathcal{D}_2} [\mathcal{A} \text{ accepts } x']\right|.
\end{equation}

To demonstrate an oracle separation between $\BQP$ and $\PH$, it suffices to identify a distribution $\mathcal{D}$ that is pseudorandom for $AC^0$ circuits, but not for efficient quantum algorithms making $\operatorname{poly}(n)$ queries.

First, we review the distribution \cite{raz2022oracle}$ \mathcal{D}$ using truncated multivariate Gaussians. Let $n \in \mathbb{N}$ and $N = 2^n$. Define $\epsilon = 1/(C \ln N)$ for a sufficiently large constant $C > 0$. Define $\mathcal{G}$ as a multivariate Gaussian distribution over $\mathbb{R}^N \times \mathbb{R}^N$ with mean $0$ and covariance matrix:
\begin{equation*}
\epsilon \cdot \begin{pmatrix}
I_N & H_N \\
H_N & I_N
\end{pmatrix},
\end{equation*}
where $H_N$ is the Hadamard transform. To take samples from $\mathcal{G}$, first sample $X = x_1, \ldots, x_N \sim \mathcal{N}_{\mathbb{R}}(0, \epsilon)$, and let $Y = H_N \cdot X$. Define $\mathcal{G}'$ as the distribution over $Z = (X, Y)$. Let $\operatorname{trunc}(a) = \min(1, \max(-1, a))$. The distribution $\mathcal{D}$ over $\{\pm 1\}^{2N}$ begins by drawing $Z \sim \mathcal{G}'$. Then for each $i \in [2N]$, draw $z'_i = 1$ with probability $\frac{1+\operatorname{trunc}(z_i)}{2}$ and $z'_i = -1$ with probability $\frac{1-\operatorname{trunc}(z_i)}{2}$. We then summarize the main theorem in \cite{raz2022oracle}.

\begin{theorem}[Oracle Separation of $\BQP$ and $\PH$ {\cite{raz2022oracle}}]\label{thm:raz}
The distribution $\mathcal{D}$ over inputs in $\{\pm 1\}^{2N}$ satisfies the following:
\begin{enumerate}
    \item There exists a quantum algorithm that makes one query to the input and runs in time $O(\log N)$, which distinguishes between $\mathcal{D}$ and the uniform distribution with an advantage of $\Omega\left(\frac{1}{\log N}\right) = \Omega\left(\frac{1}{n}\right)$.
    \item No Boolean circuit of quasipolynomial ($\operatorname{quasipoly}(N)$) size and constant depth distinguishes between $\mathcal{D}$ and the uniform distribution with an advantage better than $\frac{\operatorname{polylog}(N)}{\sqrt{N}}$.
\end{enumerate}  
Hence, the separation follows.
\end{theorem}

\cite{bassirian2021certified} defines another problem called the SquaredForrelation problem and shows that the SquaredForrelation problem can also be used to construct the oracle separation of $\mathsf{BQP}$ and $\mathsf{PH}$. Let $\mathcal{D}'$ be defined as follows. We first sample $X = x_1, \ldots, x_N \sim \mathcal{N}_{\mathbb{R}}(0, \epsilon)$, and let $Y = H_N \cdot X$. Define $\mathcal{G}'$ as the distribution over $Z = (X, Y^2 - \epsilon)$. The distribution $\mathcal{D}'$ over $\{\pm 1\}^{2N}$ first draws $Z \sim \mathcal{G}'$. Then for each $i \in [2N]$, draws $z'_i = 1$ with probability $\frac{1+\operatorname{trnc}(z_i)}{2}$ and $z'_i = -1$ with probability $\frac{1-\operatorname{trnc}(z_i)}{2}$. 

\begin{definition}[{SquaredForrelation \cite{bassirian2021certified}}] 
\label{def:SquaredForrelation}
Given oracle access to Boolean functions $f, g: \{0, 1\}^n \to \{ \pm 1\}$, distinguish whether they are sampled according to a Forrelated distribution $\mathcal{D}'$ or uniformly at random.
\end{definition}
\cite{bassirian2021certified} shows that $\mathcal{D}'$ has similar properties as $\mathcal{D}$. 
\begin{theorem}[\cite{bassirian2021certified}]
The distribution $\mathcal{D}'$ over inputs in $\{\pm 1\}^{2N}$ satisfies the following:
\begin{enumerate}
\item There exists a quantum algorithm that makes one query to the input and runs in time $O(\log N)$, which distinguishes between $\mathcal{D}'$ and the uniform distribution with an advantage of $\Omega\left(\frac{1}{\log N}\right) = \Omega\left(\frac{1}{n}\right)$.
\item No Boolean circuit of quasipolynomial ($\operatorname{quasipoly}(N)$) size and constant depth distinguishes between $\mathcal{D}'$ and the uniform distribution with an advantage better than $\frac{\operatorname{polylog}(N)}{\sqrt{N}}$.
\end{enumerate}
\end{theorem}

Using the distribution $\D$ or $\D'$, one can construct an oracle such that $\BQP \not\subset \PH$. This proof is adapted from the works of \cite{aaronson2010bqp, fefferman2012beating, raz2022oracle}.

\section{Noisy Quantum Advantage over $\NP$}

In this section, we will first recall our noisy model and provide the formal version of \thm{sh-1}, demonstrating that it implies the desired oracle separation. Afterward, we will delve into the proof details.

Here, we reintroduce our noisy model. Our noise model closely mirrors a quantum circuit but includes a depolarization channel applied after each gate, including the oracle. To solve the Deutsch-Jozsa (DJ) problem, we simply adapt the original quantum circuit of the DJ algorithm. Each gate is subjected to the depolarization channel. We summarize the detailed circuit in Figure \ref{Fig:DJE}. Here, $\textbf{H}$ represents the Hadamard gate, $\textbf{E}$ denotes the depolarization noise, and $\O_f$ denotes the oracle realization of $f$.

\begin{theorem}\label{thm:rigor-4}
There exists a constant $\lambda_c > 0$ such that for all $\lambda < \lambda_c$, a $\BQP_{\lambda}$ algorithm is capable of solving the Deutsch–Jozsa problem within a maximum of $n^2 \log(n)$ queries, achieving a probability greater than $\frac{2}{3}$.

\end{theorem}
\thm{rigor-4}  implies the following corollary.

\begin{corollary}\label{cor:rigor-4}
 There exist a constant $\lambda>0$ and an oracle $\O$ relative to which $\BQP_{\lambda}^{\O} \not\subset \NP^{\O}$.
\begin{equation}
\BQP_{\lambda}^{\O} \not\subset \NP^{\O}.
\end{equation}
Such a task can be solved using the Deutsch–Jozsa algorithm.   
\end{corollary}
\begin{proof}
\thm{rigor-4} demonstrates that we can solve the DJ problem using $\poly(n)$ queries. Conversely, \thm{bb} indicates that the DJ problem can be leveraged to construct an oracle that cannot be solved in $\NP$ by accessing it.
\end{proof}

\subsection{Proof of \thm{rigor-4}}

Here we first outline the proof steps to demonstrate the solvability of the Deutsch-Jozsa problem using the original quantum circuit even in the presence of noise. We appeal to the following two steps. First, since we assume our noise is a depolarization channel, we can demonstrate that we can view the input before the oracle as a statistical mixture of $\ket{+}$ and $\ket{-}$ states on each qubit. Second, after the oracle and the final Hadamard gate, we measure the output state and perform a majority vote. The technical challenge arises from the fact that the noise affects the measurement result, making distinguishing constant or balanced functions difficult. The compromise is that we need to repeat the process $O(n^2 \log n)$ times instead of using constant iterations. Now let us delve into the detailed steps and introduce the necessary notations.

\begin{figure}[h]
\centering
\includegraphics[scale=0.5]{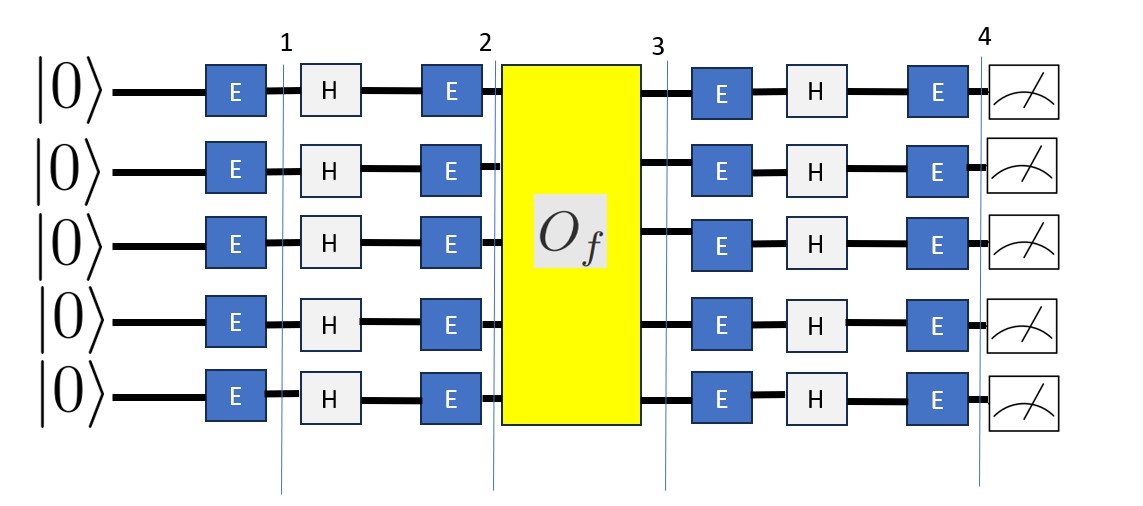}
\caption{DJ algorithm in the presence of local
noise denotes that with probability $\lambda$, an depolarization channel of single-qubit operation is
applied. We have labeled the layers of noise for ease of reference in the proof.
Here $\textbf{H}$ represents the Hadamard gate and $\textbf{E}$ represents the depolarization noise. $\O_f$ indicates the oracle realization of $f.$}
\label{Fig:DJE}
\end{figure}

Given $f:\bit^n\to \{ \pm 1\}$ and denote the Fourier spectrum of $f$ as the set $S_f:=\{s | \hat{f}(s)\neq 0\}$, where $\hat{f}(y)$ represents the Hadamard transform of $y \in \{0,1\}^{n}$. 
We now demonstrate that $\BQP_{\lambda}$ is capable of solving the Deutsch–Jozsa problem. The overarching concept is to perform $M$ measurements on the final output (with the value of $M$ to be determined later) and subsequently compute the average over the obtained bits. 
For a constant function, $S_f=\{0^n\}$; for a balanced function, $0^n\notin S_f$. In order to determine whether $f$ is constant or balanced, we need to ascertain whether $S_f$ exclusively contains $0^n$ or not. Now let us move to the first step showing that the input state before the oracle can be viewed as a statistical mixture of $\ket{+}$ and $\ket{-}$ states on each qubit. .

Consider the following single-qubit quantum circuit, starting with the initialization at $\ket{0}$ and followed by a single Hadamard gate. We can easily calculate the output density matrix.

After the initialization of $\ket{0}\bra{0}$, we have the following mixed state due to the depolarization (stage 1 in Fig \ref{Fig:DJE}):
\begin{equation}
    D_{\lambda}(\ket{0}\bra{0}) = (1-\lambda)\ket{0}\bra{0} + \lambda(\ket{0}\bra{0} + \ket{1}\bra{1}) = (1-\lambda/2)\ket{0}\bra{0} + \lambda/2\ket{1}\bra{1}
\end{equation}
After the Hadamard gate and another depolarization noise (stage 2 in Fig \ref{Fig:DJE}), we have the following mixed state:
\begin{equation}
    D_{\lambda}\left( H((1-\lambda/2)\ket{0}\bra{0} + \lambda/2\ket{1}\bra{1})H^{\dagger} \right) = D_{\lambda}( (1-\lambda/2)\ket{+}\bra{+} + \lambda/2\ket{+}\bra{-} )
\end{equation}
We can combine the results as follows:
\begin{equation}
    (1-\lambda)(1-\lambda/2)\ket{+}\bra{+} + \lambda/2\ket{-}\bra{-} + \lambda/2 (\ket{+}\bra{+} + \ket{-}\bra{-}) = p_1 \ket{+}\bra{+} + p_2\ket{-}\bra{-}
\end{equation}
Here, we use the notation $I = \ket{+}\bra{+} + \ket{-}\bra{-}$.

Thus, we can observe that each qubit, with probability $p_1$, will remain in the $\ket{+}$ state; otherwise, it becomes $\ket{-}$. We define $p_1 = (1-\lambda)(1-\lambda/2) + \lambda/2$.

For the general $n$-qubit state, given a set of error locations $E$ ($|E| = k$) with probability $p_1^k(1-p_1)^{n-k}$, the state before the oracle is $\sum_{x}\frac{1}{2^{n/2}}(-1)^{x \cdot E}\ket{x}$, and after the oracle, the state becomes $\sum_{x}\frac{1}{2^{n/2}}(-1)^{x \cdot E}f(x)\ket{x}$. Here, we view $E$ as an $n$-bit vector, where $E_i=1$ if an error occurs at the $i$-th qubit, and $0$ otherwise.

If there is no error after the oracle, the state becomes $\sum_{x}\frac{1}{2^{n/2}}(-1)^{x \cdot E}f(x)\ket{x}$. After applying another Hadamard gate, the state transforms into
$\sum_{x,y}\frac{1}{2^n}(-1)^{x \cdot E}f(x)(-1)^{x \cdot y}\ket{y}$.
If we express $f(x)$ in terms of its Fourier spectrum as $f(x) = \sum_{s}\hat{f}(s)(-1)^{s \cdot x}$, then we have
\begin{equation}
\sum_{x,y,s}\frac{1}{2^n}(-1)^{x \cdot (E+y+s)}\hat{f}(s)\ket{y} = \sum_{y,s}\delta_{s+E,y}\hat{f}(s)\ket{y}.
\end{equation} This simplifies to
\begin{equation}
\sum_{y,s}\delta_{s+E,y}\hat{f}(s)\ket{y}.
\end{equation}
The measurement outcomes must correspond to $s+E$ with a probability of $|\hat{f}(s)|^2$. Since the input state is a statistical mixture involving $E$, and as $E$ is independent of the oracle, the measurement outcomes will result in $s+E$ with a probability of $|\hat{f}(s)|^2$.

We can proceed with an analysis of the expectation value of the $i$-th bit. If $s_i=1$, then with a probability of $p_1$, the bit remains as is, and with a probability of $1-p_1$, it becomes $0$. Similarly, if $s_i=0$, then with a probability of $p_1$, the bit remains unchanged, and with a probability of $1-p_1$, it becomes $1$. The expectation value of the $i$-th bit is given by:
\begin{equation}
\Exp[X_i] = \sum_{s_i=1}|\hat{f}(s)|^2 p_1 + \sum_{s_i=0}|\hat{f}(s)|^2(1-p_1) = \sum_{s_i=1}|\hat{f}(s)|^2 p_1 + (1-p_1) \left(1 - \sum_{s_i=1}|\hat{f}(s)|^2\right).
\end{equation}
This can be further simplified as follows:
\begin{equation}
\Exp[X_i] = (2p_1-1)\sum_{s_i=1}|\hat{f}(s)|^2 + 1-p_1.
\end{equation}
Now, let's extend our analysis to include the effects of stages 3 and 4 in Fig. \ref{Fig:DJE}, both of which involve single-qubit gates. We summarize this extension in the following lemma.

\begin{lemma}\label{lem:nDJ}
For each $i \in [n]$, with a probability of $(1-\lambda)^2$, the $i$-th output bit is a binary random variable (taking values in $\{0, 1\}$) whose expectation value is given by $\Exp[X_i]$. Otherwise, the output bit is randomly chosen with equal probabilities.
\end{lemma}
\begin{proof}
Referring to Figure \ref{Fig:DJE}, when considering the $i$-th bit, with a probability of $(1-\lambda)^2$, the expectation value corresponds to $\Exp[X_i]$. Thus, if no error occurs on the $i$-th bit, the expectation value of the $i$-th bit becomes
\begin{equation}
    y_i = (1-\lambda)^2\Exp[X_i] + (1 - (1-\lambda)^2)\frac{1}{2}.
\end{equation}
For later reference, we'll denote $y_i$ as $(1-\lambda)^2\Exp[X_i] + (1 - (1-\lambda)^2)\frac{1}{2}$. Upon detailed calculation, we find
\begin{equation}
    y_i = (\lambda-1)^4\sum_{s_i=1}|\hat{f}(s)|^2 + \frac{1}{2}(-\lambda+2)\lambda((\lambda-2)\lambda+2).
\end{equation}
Here, we introduce the notation $g(\lambda) = \frac{1}{2}(-\lambda+2)\lambda((\lambda-2)\lambda+2)$, where $g(\lambda) > 0$ for $0 \leq \lambda < 1$. Additionally, let's use $\sum_{s_i=1}|\hat{f}(s)|^2 = \bar{s}_i$.
\end{proof}

\lem{nDJ} establishes the following: If $f$ is constant, then the expectation values of the output bits will be a fixed function of $\lambda$. If $f$ is balanced, meaning $0^n \notin S_f$, then the expectation values of the output bits will be equal to a fixed function of $\lambda$ plus $(1-\lambda)^4 \bar{s}$, where $\bar{s}=\sum_{\gamma \in S_f}|\hat{f}(\gamma)|^2\gamma$. Here, $\bar{s}$ must be nonzero when $f$ is balanced. However, it is conceivable that all components of $\bar{s}$ are exponentially small, rendering them challenging to discern from $0$. To circumvent the scenario where all components of $\bar{s}$ are exponentially small, we establish that at least one component of $\bar{s}$ must be greater than or equal to $\frac{1}{n}$ when $f$ is balanced. Consequently, we can perform simple statistics and compute the average to distinguish $f$. Notice we have the following simple fact.

For an integer $n\geq 1$, let $X_1,\ldots, X_n$ be $n$ binary random variables (not necessarily independent or identical), and let $X=(X_1,\ldots, X_n)$. If for every $i\in[n]$, $\Exp[X_i]<1/n$, then $\Pr[X=0^n]>0$. Using a union bound, we have:
\begin{align}
\Pr[X\neq 0^n] \leq \sum_{i=1}^n \Pr[X_i=1] < 1.
\end{align}
For balanced functions, $\Pr[X=0^n]=0$, implying that there exists an $i\in[n]$ such that $\Exp[X_i]\geq 1/n$. Additionally, we observe that the bound is tight.

We are now prepared to demonstrate that the Deutsch–Jozsa problem can be solved using a $\BQP_{\lambda}$ 
approach and show \thm{rigor-4}.
\begin{proof}
We describe our algorithm $\A$ as follows: we run the standard Deutsch–Jozsa algorithm for a total of $M$ times (where the value of $M$ will be determined later). During each run, we measure the $n$ output qubits in the computational basis, resulting in a bit string each time. We then compute the average for each $i \in [n]$. Let $Y_i$ be a random variable representing the $i$-th bit, and $\bar{Y}_i$ denote the average of $M$ independent copies of $Y_i$. According to \lem{nDJ}, we have $\Exp[\bar{Y}_i] = y_i$. The algorithm is structured as follows:
\begin{itemize}
    \item If $\forall i \in [n]$, $\bar{Y}_i \leq (\lambda-1)^4 \frac{1}{2n} + g(\lambda)$, then $\A$ outputs $1$, indicating that $f$ is constant.
    \item Otherwise, $\A$ outputs $0$, signifying that $f$ is balanced.
\end{itemize}

We can apply Hoeffding's inequality (\thm{hoe}) as follows:
\begin{equation}
    \Pr\left[\bar{Y}_i - ((\lambda-1)^4 \frac{1}{2n} + g(\lambda)) > \epsilon\right] \leq e^{-2M\epsilon^2}.
\end{equation}
If $f$ is balanced, then for some $i$, $\bar{s}_i = \sum_{s_i=1}|\hat{f}(s)|^2$ must be greater than $\frac{1}{n}$. Consequently, the probability of $\bar{Y}_i$ being below $(1-\lambda)^4 \frac{1}{2n} + g(\lambda)$ becomes small as $M$ increases. We can set $\epsilon = \frac{(1-\lambda)^4}{2n}$ and apply Hoeffding's inequality as follows:
\begin{equation}
    \Pr\left[(g(\lambda) + (1-\lambda)^4\bar{s}_i) - \bar{Y}_i > \frac{(1-\lambda)^4}{2n} \right] \leq e^{-2M \frac{(1-\lambda)^8}{4n^2}}.
\end{equation}
This implies that
\begin{equation}
    \Pr\left[\bar{Y}_i < (\lambda-1)^4 \frac{1}{2n} + g(\lambda)\right] \leq e^{-2M \frac{(1-\lambda)^8}{4n^2}}.
\end{equation}
Choosing
\begin{equation}
    M = O\left(n^2\frac{1}{(1-\lambda)^8}\right)
\end{equation}
ensures that $e^{-2M \frac{(1-\lambda)^8}{4n^2}} \leq \frac{1}{3}$. Thus, when $f$ is balanced, we have $\Pr[\bar{Y}_i < (\lambda-1)^4 \frac{1}{2n} + g(\lambda), \forall i \in [n]] \leq \frac{1}{3}$ for $M = O(n^2)$.

For the case when $f$ is constant (and hence, $\bar{s}_i = 0$ for all $i$), the probability of failure at any of the $n$ sites can be upper-bounded as follows:
\begin{equation}
    \Pr\left[\max_{i}\bar{Y}_i - g(\lambda) > \epsilon\right] \leq \sum_{i}\Pr\left[\bar{Y}_i - g(\lambda) > \epsilon\right] \leq n e^{-2M \epsilon^2}.
\end{equation}
By setting $\epsilon = \frac{(1-\lambda)^4}{2n}$ as before, we find:
\begin{equation}
    \Pr\left[\max_{i}\bar{Y}_i - g(\lambda) > \frac{(1-\lambda)^4}{2n}\right] \leq \sum_{i}\Pr\left[\bar{Y}_i - g(\lambda) > \frac{(1-\lambda)^4}{2n} \right] \leq n e^{-2M \frac{(1-\lambda)^8}{4n^2}},
\end{equation}
and by selecting $M$ such that
\begin{equation}
    M = 8n^2 \log(n),
\end{equation}
we ensure that the right-hand side is at most $\frac{1}{3}$. Thus, when $f$ is constant, we have $\Pr[\bar{Y}_i < (\lambda-1)^4 \frac{1}{2n} + g(\lambda), \forall i \in [n]] \geq \frac{2}{3}$. Consequently, we can employ $O(n^2 \log(n))$ queries to solve the Deutsch–Jozsa problem.

In summary:
\begin{itemize}
    \item If $f$ is constant, $\Pr[\bar{Y}_i \leq (\lambda-1)^4 \frac{1}{2n} + g(\lambda), \forall i \in [n]] \geq \frac{2}{3}$.
    \item If $f$ is balanced, $\Pr[\bar{Y}_i \leq (\lambda-1)^4 \frac{1}{2n} + g(\lambda), \forall i \in [n]] \leq \frac{1}{3}$.
\end{itemize}
\end{proof}

\section{Noisy Quantum Advantage over $\PH$}\label{sec:SQ}
The goal of this section is to prove \thm{sh-3}. We define our noisy setting identical to $\BQP_{\lambda}$; however, all errors occur before any interaction with the oracle $\mathcal{O}$. This assumption may seem artificial, but it is crucial for our analysis. In our scenario, there is only one query and a single classical query after the Fourier sampling. By addressing this compromise, we aim to demonstrate a weaker separation between $\BQP_{\lambda}$ and $\PH$, where $\lambda(n) = \frac{\log n}{n}$.

We also introduce a more general version of the error model by defining a probability distribution over all error vectors $e$, characterized by $\mathcal{E} : \{0,1\}^n \rightarrow [0,1]$, where $p_e = \mathcal{E}(e)$ represents the probability of error $e$ occurring. The complexity class $\BQP_{\mathcal{E}}$ is indicated by this model. Specifically, for $\BQP_{\lambda(n)}$, we have:
\[
\mathcal{E}(e) = (1-\lambda(n))^{n-|e|} \lambda(n)^{|e|}
\]
where $|e|$ denotes the Hamming weight of $e$. We assume direct access to the phase oracle without utilizing the control qubit. This setup allows us to demonstrate that the error in the control qubit does not influence our results, as discussed in Remark~\ref{rm-f}. To establish the separation, we introduce the problem known as Square Forrelation.
Here we present a framework to modify distributions inspired by \cite{bassirian2021certified}. We call it $\mathbf{SQUAREDFORRELATION} (\Sigma)$ . Let the two distributions be defined as follows:
\begin{itemize}
    \item Yes: $X \sim \N(0,\Sigma)$ and $Y = H_N \cdot X.$ Let $Y_i'=Y_i^2-\Exp[Y_i]^2$ and output $(X,Y').$
    \item No:  $X \sim \N(0,\Sigma)$ and $Y \sim \N(0,H\Sigma H)$ Let $Y_i'=Y_i^2-\Exp[Y_i]^2$ and output $(X,Y').$
\end{itemize}
Notice that if we take $\Sigma=\epsilon=O(1/n)$ which recovers \cite{raz2022oracle, bassirian2021certified}.
What \cite{bassirian2021certified} shows can be restated as follows.
There exists a large enough $c_1>0$ such that:
\begin{itemize}
    \item There is a $\BQP$ algorithm which takes 1 query with $\O(\log n)$ time solving \\
    $\mathbf{SQUAREDFORRELATION} (\frac{1}{c_1n}\Id).$
    \item  For any $AC^{0}$ circuit can not solve \\
    $\mathbf{SQUAREDFORRELATION} (\frac{1}{c_1n}\Id).$
\end{itemize}
such statement can be re-intepretated as the oracle separation between $\BQP$ and $\PH$ \cite{aaronson2010bqp}.

By modifying $\Sigma,$ we can deduce the other kind of Separation. We use the following simple fact or Linear Transformation of Gaussian Random Variable.
\begin{equation}
X \sim \N(\mu, \Sigma), HX \sim \N(H\mu, H\Sigma H^{T})
\end{equation}
In terms of the covariance matrix, in the Yes case, the covariance matrix: \begin{equation}
    (X,Y) \sim \begin{pmatrix}
\Sigma & \Sigma H \\
H\Sigma^{T} & H \Sigma H
\end{pmatrix}
\end{equation}
In the No case, the covariance matrix: \begin{equation}
    (X,Y) \sim \begin{pmatrix}
\Sigma & 0\\
0 & H \Sigma H
\end{pmatrix}
\end{equation}.
The goal is to distinguish which is the case with constant successful probability. Notice that We draw $(X,Y')$ and truncate each $(X,Y')_i$ to the interval $[-1, 1]$ by applying the function $\trnc(a)=\min(1,\max(-1, a)).$  Then, in order to obtain values in $\{ \pm 1 \}$ independently
for each $i \in [2N]$ we assign it to $1$ with probability $\frac{1+\trnc(X_i)}{2}$ or $\frac{1+\trnc(Y'_i)}{2}$)
and $-1$ with probability
$\frac{1+\trnc(X_i)}{2}$ or $\frac{1+\trnc(Y'_i)}{2}$. We output $(X,Y') \in \{ \pm 1 \}^{2N}.$ Let's assume there is no error happening after $O_{f}.$ We will calculate the probability using the None-truncated version and then show the difference between the truncated version and the None-truncated version is very close as \cite{raz2022oracle, bassirian2021certified}.

We denote $\mathfrak{G}, \mathfrak{U}$ to be the distribution for yes/no instances, and $\mathscr{G}, \mathscr{U}$ to be their truncated versions. We use the following diagram to indicate their relationship and provide the proof steps.

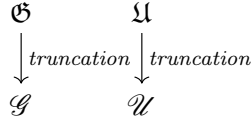
\begin{figure}[h]
    \centering
    \begin{tikzcd}
    \mathfrak{G} \arrow[d, "\textit{truncation}"] & \mathfrak{U} \arrow[d, "\textit{truncation}"] \\
    \mathscr{G} & \mathscr{U}
    \end{tikzcd}
    \caption{Illustration of relationship between distribution}
\end{figure}

\begin{enumerate}
    \item We will choose a particular $\Sigma$ such that there exists a $\BQP_{\mathcal{E}}$ algorithm $Q$ making one query and running in time $O(\log N)$ such that 
    \[
    \Exp_{u \sim \mathscr{G}}[Q(u)] - \Exp_{u \sim \mathscr{U}}[Q(u)] \geq \frac{1}{\poly(n)}.
    \]

    \item We show the previous statement by demonstrating that the truncation error can be neglected for both $\mathfrak{G}$ and $\mathfrak{U}$. We will show that for any multilinear function $f$:
    \[
    \left| \Exp_{u \sim \mathfrak{G}}[f(u)] - \Exp_{u \sim \mathscr{G}}[f(u)] \right| \leq \frac{1}{4N^c}, \quad \left| \Exp_{u \sim \mathfrak{U}}[f(u)] - \Exp_{u \sim \mathscr{U}}[f(u)] \right| \leq \frac{1}{4N^c}.
    \]
    for some $c>0.$

    \item Finally, we show the new classical hardness: Let $A: \{\pm 1\}^{2N} \to \{\pm 1\}$ be a Boolean circuit of size $\exp(\log^{O(1)}(N))$ and depth $O(1)$. Then
    \[
    \left| \Exp_{u \sim \mathscr{G}}[A(u)] - \Exp_{u \sim \mathscr{U}}[A(u)] \right| \leq \frac{\polylog(N)}{N^c}
    \]
    for some $c > 0$.
\end{enumerate}

\begin{theorem}\label{thm:summary}
Here we present our main results inspired by \cite{bassirian2021certified}. Consider the problem $\mathbf{SQUAREDFORRELATION} (\Sigma)$, where $\Sigma$ is an $N \times N$ covariance matrix and $H_N$ is the standard $N \times N$ Hadamard matrix. Define two distributions as follows:
\begin{itemize}
    \item Yes: $X \sim \mathcal{N}(0,\Sigma)$ and $Y = H_N \cdot X$. Let $Y_i' = Y_i^2 - \mathbb{E}[Y_i^2]$ and output $(X,Y')$.
    \item No:  $X \sim \mathcal{N}(0,\Sigma)$ and $Y \sim \mathcal{N}(0, H_N \Sigma H_N)$. Let $Y_i' = Y_i^2 - \mathbb{E}[Y_i^2]$ and output $(X,Y')$.
\end{itemize}
Define $\tilde{\Sigma} = H_N \Sigma H_N$. We are also given a class $\BQP_{\mathcal{E}}$.

\begin{enumerate}
    \item \textbf{Completeness:} If $\lambda_{\max}(\Sigma) \leq \frac{1}{\mathrm{poly}(n)}$ and 
    \begin{equation}
      \frac{2}{N}\sum_{e \in \{0,1\}^n}\mathcal{E}(e)\sum_{i \in \{0,1\}^n}\tilde{\Sigma}_{i,i+e}^2 \geq \frac{1}{\mathrm{poly}(n)}
    \end{equation}
    then there exists a $\BQP_{\mathcal{E}}$ algorithm that can distinguish the two distributions, namely\\
    $\mathbf{SQUAREDFORRELATION} (\Sigma)$ is in $\BQP_{\mathcal{E}}$.

    \item \textbf{Soundness:} If $\|H_N \Sigma\|_{\infty} \leq O\left(\frac{\epsilon_0}{N^m}\right)$ for some $\epsilon_0 = O\left(\frac{1}{n}\right)$ and $\max \{\|\Sigma\|_{\infty}, \|H_N \Sigma H_N\|_{\infty}\} \leq \epsilon_0$ and $1/2 \geq m > 0$, then for any $A: \{\pm 1\}^{2N} \to \{\pm 1\}$, a Boolean circuit of size $\exp(\log^{O(1)}(N))$ and depth $O(1)$, 
    \begin{equation}
    \left| \mathbb{E}_{z \sim \text{yes}}[A(x,y')] - \mathbb{E}_{z \sim \text{no}}[A(x,y')] \right| 
    \leq \frac{\mathrm{polylog}(N)}{N^m}.        
    \end{equation}

\end{enumerate}
If such $\Sigma$ exists and satisfies all the requirements above with known techniques \cite{raz2022oracle, aaronson2010bqp}, then we can use such $\Sigma$ to construct a decision problem and an oracle $\mathcal{O}$ such that this problem cannot be solved in $\PH^{\mathcal{O}}$ but can be solved using $\BQP_{\mathcal{E}}^{\mathcal{O}}$ with access to the same oracle.
\end{theorem}

Now we can use it to prove \thm{sh-3}:
\begin{corollary}
By choosing $\Sigma = \frac{1}{c_1 n}$ for a small enough $c_1 > 0$, there exists a function $\lambda: \mathbb{N} \to \mathbb{R}$ where $\lambda(n) = \Omega(\log n / n)$, and $\mathcal{O}$ relative to which $\BQP_{\lambda} \not\subset \PH$.
\end{corollary}

\begin{proof}
We check these two conditions one by one. Since $\Sigma = \frac{1}{c_1 n}$, we have $\lambda_{\max}(\Sigma) = \frac{1}{c_1 n} \leq \frac{1}{\mathrm{poly}(n)}$. We now have:
\begin{equation}
  \frac{2}{N} \sum_{e \in \{0,1\}^n} \mathcal{E}(e) \sum_{i \in \{0,1\}^n} \tilde{\Sigma}_{i,i+e}^2 \geq \frac{2}{N} \sum_{e = \{0,1\}^n} \mathcal{E}(0^n) \sum_{i \in \{0,1\}^n} \tilde{\Sigma}_{i,i}^2 = \frac{2 (1 - \lambda(n))^n}{c_1^2 n^2} \approx \Omega\left(\frac{1}{n^3}\right).
\end{equation}
Here we use $\mathcal{E}(0^n) = (1 - \lambda(n))^n.$
Furthermore, it is trivial to see that
\[
(1 - \log n / n)^n = \frac{1}{n} + o\left(\frac{1}{n^2}\right).
\]
Hence condition (1) is satisfied. Condition (2) is a simple calculation by choosing $m = \frac{1}{2}.$
\end{proof}

One might question the need to introduce this formalism by modifying the standard Forrelation \cite{bassirian2021certified} to $\Sigma$. Although demonstrating separation when the error rate is $\Omega(\log n / n)$ does not require modifying this problem for other noise models, such a framework can still prove useful. Moreover, it introduces a new aspect of classical hardness that has intrinsic interest. Below, we address the limitations of this argument.
\paragraph{Limit of Square Forrelation}
Now we show the following no-go theorem, claiming that \thm{summary} cannot provide separation when \(\lambda = \lambda_0\) is a constant. Essentially, we will demonstrate that there is no valid covariance matrix \(\Sigma\) that satisfies condition (1) in \thm{summary}.
First, we present the following no-go theorem in matrix analysis:
\begin{lemma}
There is no symmetric positive semi-definite (PSD) \(N \times N\) matrix \(C\) such that \(\sum_{ij} c_{ij}^2 \geq N \exp(n)\) and the largest eigenvalue \(\leq \poly(n)\).
\end{lemma}
\begin{proof}
Since \(C\) is symmetric, we can use its spectral decomposition:
\[ C = \sum_{k=1}^{N} \lambda_k u_k u_k^T, \]
where \(\lambda_k\) are the non-zero eigenvalues, and \(u_k\) are the corresponding unit norm eigenvectors. We can write:
\begin{align}
c_{ij} &= \sum_{k} \lambda_k u_{ki} u_{kj}, \\
c_{ij}^2 &= \sum_{k,l} \lambda_k \lambda_l u_{ki} u_{kj} u_{li} u_{lj} \implies \sum_{ij} c_{ij}^2 = \sum_{k,l} \lambda_k \lambda_l \sum_{i} u_{ki} u_{li} \sum_{j} u_{kj} u_{lj} = \sum_{k} \lambda_k^2.
\end{align}
Note that, by the spectral theorem, \(u_{ki}\) is an \(N \times N\) unitary matrix, hence \(\sum_{i} u_{ki} u_{li} = \delta_{kl}\). Consequently, we have:
\begin{align}
N \exp(n) \leq \sum_{ij} c_{ij}^2 \leq N \lambda_{\max}^2,
\end{align}
which implies \(\lambda_{\max} \geq \exp(n).\)
\end{proof}

This shows that if \(\lambda_{\max}(n) \leq \poly(n)\), then it implies that \(\sum_{ij} c_{ij}^2 \leq N \poly(n)\). Now, we can write the following equation:
\begin{align*}
  \frac{2}{N}\sum_{e \in \{0,1\}^n}\mathcal{E}(e)\sum_{i \in \{0,1\}^n}\tilde{\Sigma}_{i,i+e}^2 = \frac{2}{N}\sum_{k=0}^{n}\lambda_0^{k}(1-\lambda_0)^{n-k}\sum_{e \in \{0,1\}^n, |e|=k}\sum_{i \in \{0,1\}^n}\tilde{\Sigma}_{i,i+e}^2.
\end{align*}
Without loss of generality, we can assume \(0 < \lambda_0 \leq \frac{1}{2}\) which implies \(\lambda_0^{k}(1-\lambda_0)^{n-k} \leq (1-\lambda_0)^n\). Now we have:
\begin{align*}
& \frac{2}{N}\sum_{k=0}^{n}\lambda_0^{k}(1-\lambda_0)^{n-k}\sum_{e \in \{0,1\}^n, |e|=k}\sum_{i \in \{0,1\}^n}\tilde{\Sigma}_{i,i+e}^2 \\
& \leq \frac{2(1-\lambda_0)^{n}}{N}\sum_{k=0}^{n}\sum_{e \in \{0,1\}^n, |e|=k}\sum_{i \in \{0,1\}^n}\tilde{\Sigma}_{i,i+e}^2 \\
& \leq \frac{2(1-\lambda_0)^{n}}{N}\sum_{e \in \{0,1\}^n}\sum_{i \in \{0,1\}^n}\tilde{\Sigma}_{i,i+e}^2 \\
& \leq 2(1-\lambda_0)^{n} \poly(n).
\end{align*}
This expression will always be exponentially small in terms of \(n\) if \(\lambda_0\) is a constant.

We remark that in this section, we only consider the Fourier sampling algorithm which makes only one query to $f$. 
After obtaining the samples, the algorithm makes one query to $g$ to determine the decision. 
In this setting, we may view the problem as the task of distinguishing the Fourier spectrum of a Forrelated $f$ affected by bit-flipping noise with constant probability $p$ from the uniform distribution using the oracle $g$. 
When the error rate $p=0$, Bassirian, Bouland, Fefferman, Gunn, and Tal showed that $g$ serves as a good distinguisher that determines which case with inverse polynomial advantage \cite{bassirian2021certified}. 

Although we cannot prove the hardness even against one-query algorithms, we observe that there is an obstacle toward proving the existence of a distinguisher given oracle access to $g$ and the samples. 
Indeed, the existence of a distinguisher only making oracle access to $g$ may potentially lead to a black-box reduction for solving Learning Parity with Noise (LPN) in the regime of constant noise rate to the error-free case, which is known to be easy. \ 
The problem is believed to be hard for quantum algorithms \cite{pietrzak2012cryptography}.
We believe that any algorithm that distinguishes the samples would require substantial uses of the structure from Forrelation.

Now let us move to the completeness proof, namely the quantum easiness part in \thm{summary}.
\subsection{Quantum Easiness}
We use the following algorithm $A$ to distinguish between functions sampled from $\mathscr{G}$ and $\mathscr{U}$:
\begin{enumerate}
    \item Apply the Fourier transform on $f$ and sample $x$ from the distribution induced by $\hat{f}$.
    \item Accept if $g(x) = 1$, and reject otherwise.
\end{enumerate}
Notice that this algorithm only uses one quantum oracle. We discuss the criterion for distinguishing the above two distributions using the noisy quantum model. Let $\tilde{\Sigma} = H \Sigma H$. 

\begin{theorem}\label{thm:nisq}
When error $e$ occurs at initialization, we have:
\begin{equation}
\mathbb{E}_{\text{Yes}}[\psi_e(X,Y')] - \mathbb{E}_{\text{No}}[\psi_e(X,Y')] = \frac{1}{N}\sum_{i \in \{0,1\}^n} 2\tilde{\Sigma}_{i,i+e}^2.
\end{equation}
\end{theorem}

\begin{proof}
For an error vector $e \in \{0,1\}^n$, define $\psi_e(X,Y') = \frac{1}{N}\sum_i \left(\sum_j H_{ij}X_j\right)^2 Y'_{i+e}$, where $N = 2^n.$ We use the following simple facts:
\begin{itemize}
    \item $\mathbb{E}[X_i^2] = \Sigma_{ii},$
    \item $\mathbb{E}[X_i^4] = 3\Sigma_{ii}^2,$
    \item $\mathbb{E}[X_iX_j] = 2\Sigma_{ij}^2 + \Sigma_{ii}\Sigma_{jj}.$
\end{itemize}
Now we have
    \begin{itemize}
     \item  For the No case, $X$ and $Y'$
are independent, and $\Exp[Y'_i]=0$. This
implies that $\Exp_{No}[\psi_{e}(X,Y')]=0.$
\item  For the Yes case:
\begin{equation}
   \psi_{e}(X,Y')= \frac{1}{N}\sum_{i}Y_i^2Y_{i+e}^2-\frac{1}{N}\sum_{i}Y_i^{2} 
 \tilde{\Sigma}_{i+e,i+e}.
\end{equation}
We  have 
\begin{equation}
   \Exp_{yes}\psi_{e}(X,Y')= \frac{1}{N}(\sum_{i}(2\tilde{\Sigma}_{i,i+e}^2+\tilde{\Sigma}_{i,i}\tilde{\Sigma}_{i+e,i+e}-\tilde{\Sigma}_{ii}
\tilde{\Sigma}_{i+e,i+e})=
\frac{1}{N}\sum_{i}2\tilde{\Sigma}_{i,i+e}^2 
\end{equation}
\end{itemize}
\end{proof}

We also show that the truncation error can be neglected in the above cases when $\lambda_{\max}(\Sigma) \leq \frac{1}{\mathrm{poly}(n)}$. We first derive a simple bound 
on the probability $\Pr[\forall i, z_i \in [-1,1]]$. Let $z \sim \mathcal{N}(0, \Sigma)$.
We assume $\epsilon=\Sigma_{ii}$ for all $i$.
We have
\begin{align*}
    \Pr[z_i > 1 \cup z_i < -1] &= \frac{2}{\sqrt{2\pi \epsilon}}\int_{1}^{\infty}e^{-z_i^2/(2 \epsilon)}dz_i \\
    &\leq \frac{2}{\sqrt{2\pi \epsilon}}\int_{1}^{\infty}z_ie^{-z_i^2/(2 \epsilon)}dz_i  \\
    &\leq \sqrt{\frac{2\epsilon}{\pi}}e^{-1/(2\epsilon)} \leq e^{-1/(2\epsilon)}.
\end{align*}
By the union bound, we have $\Pr[\forall i, z_i \in [-1,1]] \geq 1 - 2Ne^{-1/(2\epsilon)} \geq 1 - 1/N^c$ for any $c>0$ provided that $\epsilon = 1/(c_1n)$ for large enough $c_1$.
Next, we introduce the multivariate Gaussian tail bound.

\begin{theorem}
Let $(X_1, \dots, X_N) \sim \mathcal{N}(0, \Sigma)$.
Let $M = \Sigma^{-1}$ and $\forall i, C_i > 0$, we have 
\begin{equation}
    \Pr[X_1 > C_1 \land \dots \land X_N > C_N] \leq \frac{(\lambda_{\max}(\Sigma))^N}{(\lambda_{\min}(\Sigma))^{N/2}} e^{-\frac{1}{2} \frac{||C||^2}{\lambda_{\max}(\Sigma)}}.
\end{equation}
\end{theorem}
\begin{proof}
We compute the probability from the integral of the probability distribution:
\begin{align*}
    \Pr[X_1 > C_1 \land \dots \land X_N > C_N] &\leq \frac{|M|^{1/2}}{(2\pi)^{N/2}} \int_{C_1}^\infty \cdots \int_{C_N}^\infty \exp\left[-\frac{1}{2} X^\top M X \right] dX_1 \cdots dX_N \\
    &\leq \prod_{i=1}^{N} \int_{C_i}^\infty e^{-\frac{1}{2} \frac{X_i^2}{\lambda_{\max}(\Sigma)}}dX_i \\
    &= \prod_{i=1}^{N} \lambda_{\max}(\Sigma) e^{-\frac{1}{2} \frac{C_i^2}{\lambda_{\max}(\Sigma)}} \\
    &\leq (\lambda_{\max}(\Sigma))^N e^{-\frac{1}{2} \frac{||C||^2}{\lambda_{\max}(\Sigma)}}.
\end{align*}
Here $||C||^2 = \sum_{i=1}^{N} C_i^2$.
\begin{equation}
    |M|^{1/2} (\lambda_{\max}(\Sigma))^N \leq \frac{(\lambda_{\max}(\Sigma))^N}{(\lambda_{\min}(\Sigma))^{N/2}}.
\end{equation}
In addition, if $\frac{(\lambda_{\max}(\Sigma))^N}{(\lambda_{\min}(\Sigma))^{N/2}} < 1$, we have 
\begin{equation}
    \Pr[X_1 > C_1 \land \dots \land X_N > C_N] \leq e^{-\frac{1}{2} \frac{||C||^2}{\lambda_{\max}(\Sigma)}}. 
\end{equation}
\end{proof}

We can use a similar proof as in \cite{raz2022oracle, bassirian2021certified} to show that the truncation error can be neglected.
\begin{theorem}\label{thm:tr}
If all eigenvalues of $\tilde{\Sigma}$ (or $\Sigma$) are bounded above by $1/\poly(n)$, and $(\lambda_{\max}(\Sigma))^N / (\lambda_{\min}(\Sigma))^{N/2} < 1,$ then the truncation procedure will work effectively. Here, we write $\lambda_{\max}(\Sigma) \leq \frac{1}{c_4 n^{c_3}}$ for some $c_3 > 0, c_4 > 0.$ Notice that since $H$ is unitary, the eigenvalues of $\tilde{\Sigma}$ are the same as those of $H\tilde{\Sigma}H.$ The truncation results in:
\begin{align}
& \Exp_{(x, y)\sim \mathfrak{G}}\left[\prod_{i=1}^{N}\max(1,|x_i|)\prod_{i=1}^{N}\max(1,|y_i'|)\mathbf{1}_{(x,y)\neq (\text{trnc}(x),\text{trnc}(y))}\right] \leq \frac{1}{N^c}, \\
& \Exp_{(x, y)\sim \mathfrak{U}}\left[\prod_{i=1}^{N}\max(1,|x_i|)\prod_{i=1}^{N}\max(1,|y_i'|)\mathbf{1}_{(x,y)\neq (\text{trnc}(x),\text{trnc}(y))}\right] \leq \frac{1}{N^c}, 
\end{align}
for some $c > 0.$
\end{theorem}

\begin{proof}
Let $z = (x,y) \sim \mathfrak{G}$. For every sequence of non-negative integers $a = (a_1, \dots, a_N)$, we consider the event $\forall i \in [N] : a_i \leq |x_i| \leq a_{i+1}$, denoted by $E_a$. For every sequence of non-negative integers $b = (b_1, \dots, b_N)$, we consider the event $\forall i \in [N] : b_i \leq |y_i'| \leq b_{i+1}$, denoted by $E_b$. We have:
\begin{equation}
\Pr[E_a] \leq e^{-c_4 \frac{1}{2} n^{c_3} \sum_{i=1}^{N} a_i^2}.
\end{equation}
$\Pr[E_b] = \Pr[\forall i, b_i + \Exp[y_i^2] \leq y_i^2] \leq \Pr[\forall i, b_i \leq y_i^2] \leq 2^N \Pr[\forall i, \sqrt{b_i} \leq y_i] \leq e^{-c_4 n^{c_3} \sum_{i=1}^{N} b_i}.$ We have (notice $y_i' = y_i - \Exp[y_i^2]$):
\begin{align}
& \Exp_{(x, y) \sim \mathfrak{G}}\left[\prod_{i=1}^{N}\max(1,|x_i|)\prod_{i=1}^{N}\max(1,|y_i'|)\mathbf{1}_{(x,y) \neq (\text{trnc}(x), \text{trnc}(y'))}\right] \\
& \leq \sum_{a \in \mathbb{N}^N, b \in \mathbb{N}^N, (a,b) \neq 0^{2N}} \Pr[E_a \cap E_b] \prod_{i=1}^{N}(1+a_i)\prod_{i=1}^{N}(1+b_i) \\
& \leq \sum_{a \in \mathbb{N}^N, b \in \mathbb{N}^N, (a,b) \neq 0^{2N}} \sqrt{\Pr[E_a] \Pr[E_b]} \prod_{i=1}^{N}(1+a_i)\prod_{i=1}^{N}(1+b_i) \\
& \leq \sum_{J=1}^{\infty} e^{-2c_4 J n^{c_3}} |\{(a,b) : a \in \mathbb{N}^N, b \in \mathbb{N}^N, (a,b) \neq 0^{2N}, \sum_{i=1}^{N}(a_i + b_i) = J\}| \\
& \leq \sum_{J=1}^{\infty} e^{-2c_4 J n^{c_3}} (2N)^J \leq \sum_{J=1}^{\infty} N^{-3J} (2N)^J \leq 4N^{-2}.
\end{align}
A similar proof works for $\mathfrak{U}$ since our proof does not use any dependence between $X$ and $Y$. Notice that we can pick $c_4$ to be arbitrarily large to get any $1/N^c$ for any $c > 0.$
\end{proof}

Finally, we use the following simple facts \cite{raz2022oracle}. Let $F: \mathbb{R}^{2N} \to \mathbb{R}$ be a multilinear function that maps $\{\pm 1\}^{2N}$ to $[-1, 1]$. Let $z = (z_1, \dots, z_{2N}) \in \mathbb{R}^{2N}$. Then, $|F(z)| \leq \prod_{i=1}^{2N} \max(1, |z_i|)$. We then use the following lemma to show the truncation error can be neglected.

\begin{lemma}[{\cite[Claim~5.3]{raz2022oracle}}]\label{lem:r2}
For the same previous $\Sigma$, let $0 \leq p, p_0$ such that $p + p_0 \leq 1$. Let $F: \mathbb{R}^{2N} \to \mathbb{R}$ that maps $\{\pm 1\}^{2N}$ to $[-1, 1]$ and $F$ is multilinear. Let $z_0 \in [-p_0, p_0]^{2N}$, then:
\[
 \Exp_{z \sim \mathfrak{G}}[|F(\text{trnc}(z_0 + p \cdot z)) - F(z_0 + p \cdot z)|] \leq \frac{1}{8N^2}.
\]
\end{lemma}

\begin{proof}
The original proof works without any change with the modified Claim 5.3.
\begin{align*}
    \Exp_{z \sim \mathfrak{G}}[|F(\text{trnc}(z_0 + p \cdot z)) - F(z_0 + p \cdot z)|] &\leq \Exp_{z \sim \mathfrak{G}}[(1 + |F(z_0 + p \cdot z)|) \cdot \mathbf{1}_{\text{trnc}(z_0 + p \cdot z) \neq z_0 + p \cdot z}] \\
    &\leq \Exp_{z \sim \mathfrak{G}}[2 \cdot \prod_{i=1}^{2N} \max(1, |(z_0)_i + p \cdot z_i|) \cdot \mathbf{1}_{\text{trnc}(z) \neq z}].
\end{align*}
Finally, we use $\prod_{i=1}^{2N} \max(1, |(z_0)_i + p \cdot z_i|) \leq \prod_{i=1}^{2N} \max(1, |z_i|)$.
\begin{equation}
     \Exp_{z \sim \mathfrak{G}}[|F(\text{trnc}(z_0 + p \cdot z)) - F(z_0 + p \cdot z)|] \leq \Exp_{z \sim \mathfrak{G}}[2 \cdot \prod_{i=1}^{2N} \max(1, |z_i|) \cdot \mathbf{1}_{\text{trnc}(z) \neq z}] \leq \frac{1}{4N^2}.
\end{equation}
\end{proof}
So we have shown that the truncate error can be neglected.

\subsection{Classical Hardness}

We now describe the sufficient condition for $\mathbf{SQUAREDFORRELATION} (\Sigma)$ to be classically hard. 

For two vectors $R, Q \in \mathbb{R}^{2N}$, we denote by $R \circ Q \in \mathbb{R}^{2N}$ their point-wise product, that is $(R \circ Q)_i = R_i \cdot Q_i$ for all $i \in [2N]$.

In this section, we will generalize the results from \cite{raz2022oracle} in a more comprehensive manner. Namely, we will give sufficient conditions for a general $\Sigma$ such that distinguishing $\mathfrak{G}$ and $\mathscr{U}$ is hard for a constant depth $AC^0$ circuit. The main goal of this section is to show \thm{mainUD}. In \cite{bassirian2021certified}, we have the standard example.

\begin{theorem}\label{thm:raz}
Let $\Sigma = \frac{\Id}{c_1 n}$ for a sufficiently large $c_1$. Let $A: \{ \pm 1\}^{2N} \to \{\pm 1\}$ be a Boolean circuit of size $\exp(\log^{O(1)}(N))$ and depth $O(1)$. Then the absolute difference between the expectations $|\Exp_{z \sim \mathfrak{G}}[A(x,y')] - \Exp_{z \sim \mathfrak{U}}[A(x,y')]|$ is at most $\frac{\polylog(N)}{\sqrt{N}}$. When the truncation error can be neglected, $|\Exp_{z \sim \mathscr{G}}[A(x,y')] - \Exp_{z \sim \mathscr{U}}[A(x,y')]|$ is also at most $\frac{\polylog(N)}{\sqrt{N}}$.
\end{theorem}

\begin{theorem}\label{thm:mainUD}
Assume $\Sigma$ is positive semidefinite and satisfies $|\Sigma H|_{\infty} \leq \frac{C \epsilon_0}{N^m}$ for some $m > 0$ and $\max \{ |\Sigma_{ij}|, |(H\Sigma H)_{ij}|\} \leq \epsilon_0$ for some $C > 0$, where $\epsilon_0 = \frac{1}{c_1 n}$ for a sufficiently large $c_1$. Let $A: \{ \pm 1\}^{2N} \to \{\pm 1\}$ be a Boolean circuit of size $\exp(\log^{O(1)}(N))$ and depth $O(1)$, then $|\Exp_{z \sim \mathfrak{G}}[A(x,y')] - \Exp_{z \sim \mathfrak{U}}[A(x,y')]|$ is at most $\frac{\polylog(N)}{N^m}$. Here we assume $0< m \leq \frac{1}{2}$.
\end{theorem}

We mention some useful lemmas.
\begin{lemma}[{\cite[equation (2) in section 5]{raz2022oracle}}]\label{lem:r1}
Given a multilinear function $F: \mathbb{R}^{2N} \to \mathbb{R}$ that maps $[-1, 1]^{2N}$ to $[-1, 1]$, it has a similar expectation under both $\mathfrak{G}$ and $\mathscr{G}$.
\end{lemma}

We can apply the following theorem without modification:
\begin{equation}
\Exp_{z' \sim \mathscr{G}}[F(z')] = \Exp_{z' \sim \mathfrak{G}}[F(\text{trnc}(z'))].
\end{equation}

\begin{lemma}[\cite{tal2017tight}]
There exists a universal constant $c > 0$ such that the following holds. Let $A: \{\pm 1\}^{2N} \to \{\pm 1\}$ be a Boolean circuit with at most $s$ gates and depth at most $d$. Then, for all $k \in \mathbb{N}$, we have 
\[
\sum_{S \subset [2N] : |S| = k} |\hat{A}(S)| \leq (c \cdot \log s)^{(d-1)k}.
\]
\end{lemma}

\begin{theorem}[\cite{chattopadhyay2019pseudorandom}, Claim A.5]\label{thm:chat}
Let $f$ be a multi-linear function on $\mathbb{R}^{N}$ and $x \in [-1/2, 1/2]^{N}$. There exists a distribution over random restrictions $R_x$ such that for any $y \in \mathbb{R}^{N}$,
\[
f(x+y) - f(y) = \Exp_{\rho \sim R_x}[f_{\rho}(2 \cdot y) - f_{\rho}(0)].
\]
This implies that 
\[
f(x+y) - f(x+y') = \Exp_{\rho \sim R_x}[f_{\rho}(2 \cdot y) - f_{\rho}(2 \cdot y')].
\]
\end{theorem}

Before we start, we provide a useful bound on the moments of $\mathfrak{G}$ and $\mathfrak{U}$ by \thm{is}. We aim to upper bound the following expression, where $S$ and $T$ are multisets that may contain duplicate elements. Assuming we have good control over pair correlations, we say $\Exp[x_i y_j] \leq \frac{1}{N^m}$ and $\Exp[x_i x_j] \leq \epsilon$, $\Exp[y_i y_j] \leq \epsilon$. \thm{is} should provide a bound on such a problem.

Here we notice that $\Exp_{(x, y) \sim \mathfrak{G}}\left[\prod_{i \in S} x_i\right] \Exp_{(x, y) \sim \mathfrak{G}}\left[\prod_{j \in T} y_j\right] = \Exp_{(x, y) \sim \mathfrak{U}}\left[\prod_{i in S} x_i \prod_{j \in T} y_j\right].$
\begin{lemma}
   \begin{equation}
    \left| \Exp_{(x, y) \sim \mathfrak{G}}\left[\prod_{i \in S} x_i \prod_{j \in T} y_j\right] - \Exp_{(x, y) \sim \mathfrak{G}}\left[\prod_{i \in S} x_i\right] \Exp_{(x, y) \sim \mathfrak{G}}\left[\prod_{j \in T} y_j\right] \right| \leq 2(s+t)!\frac{\epsilon^{(s+t)/2}}{N^m}.
   \end{equation}
Here, we denote $|S| = s$ and $|T| = t$.
\end{lemma}

\begin{proof}
So now we can apply Isserlis' theorem \ref{thm:is}. Suppose we categorize the elements of sets $S$ and $T$ into three groups: $x$ paired with $x$, $y$ paired with $y$, and $x$ paired with $y$. For the term to be non-zero, we must have at least one pair of $x$ and $y$. Suppose we have $w$ pairs of $x$ and $y$, and each pair contributes $\frac{1}{N^m}$. The remaining $x$ and $y$ pairs contribute at most $\epsilon$.

Here we define $f(p) = \frac{p!}{2^{p/2} (p/2)!}$. We have:
\begin{align*}
    &\left| \mathfrak{G}(S,T) - \Exp_{(x, y) \sim \mathcal{U}}\left[\prod_{i \in S}x_i\right]\Exp_{(x, y) \sim \mathcal{U}}\left[\prod_{j \in T}y_j\right] \right| \nonumber \\
    &\leq \sum_{w=1}^{\min(|S|, |T|)} \binom{|S|}{w} \binom{|T|}{w} \frac{\epsilon^{|S|/2 + |T|/2 - w} \epsilon^w w!}{N^{mw}} f(|S|-w) f(|T|-w)
\end{align*}

Let $|S| = s$ and $|T| = t$, we have:
\begin{align*}
     & \binom{s}{w} \binom{t}{w} \frac{\epsilon^{s/2 + t/2 - w} \epsilon^w w!}{N^{mw}} f(s-w) f(t-w)  \nonumber \\
     & \leq \frac{s! t! \epsilon^{(s+t)/2}}{w! N^{mw} 2^{(s+t)/2-w} (s/2-w/2)! (t/2-w/2)!} \nonumber \\
    & \leq \frac{(s+t)! \epsilon^{(s+t)/2}}{N^{mw}}. 
\end{align*}

We now have:
\begin{equation}
\sum_{w=1}^{\min(s, t)} \frac{(s+t)! \epsilon^{(s+t)/2}}{N^{mw}} \leq 2(s+t)! \frac{\epsilon^{(s+t)/2}}{N^m}.
\end{equation}
\end{proof}

The key tool is the following:
\begin{theorem}\label{thm:f7}
Let \( F: \mathbb{R}^{2N} \to \mathbb{R} \) be a multilinear function where for all random restrictions \( L_{1,2}(F_\rho) \leq l \). Then
\begin{align*}
     (*) &= \Exp_{z = (\Delta x, \Delta y) \sim \mathfrak{G}}[F(x + \Delta x, (y + \Delta y)^2 - i H \Sigma H / t)] \\
     &\quad - \Exp_{z = (\Delta x, \Delta y) \sim \mathfrak{U}}[F(x + \Delta x, (y + \Delta y)^2 - i H \Sigma H / t)] \\
      &\leq O\left(\frac{\epsilon l}{t N^m}\right).
\end{align*}
Here, \( \Delta x \sim \mathcal{N}(0, \Sigma/t) \) and \( \Delta y = H \Delta x \).
\end{theorem}

\begin{proof}
We can express the $(*)$ value using random restrictions of $F$ applying \thm{chat}. We aim to prove the following:
Let \( t = N^{c} \) for a sufficiently large universal constant \( c \), and let \( (x,y) \) be in \( [-1/2, 1/2]^{2N} \) where \( L_{1,2}(F_\rho) \leq l \):
\begin{equation}
    (*) = \Exp_{z = (x,y) \sim \mathfrak{G}, \rho}[F_{\rho}(2 \Delta x, 2 \Delta y')] - \Exp_{z = (x,y) \sim \mathfrak{U}, \rho}[F_{\rho}(2 \Delta x, 2 \Delta y')].
\end{equation}
Here, \( \Delta y' = 2 y \Delta y + \Delta y^2 - \frac{\Sigma}{t} \), where \( x(i) = x(i-1) + \Delta x \), \( y(i) = y(i-1) + \Delta y \), and \( \Delta x \sim \mathcal{N}(0, \Sigma/t) \) and \( \Delta y = H \Delta x \) when \( \mathfrak{G} \); otherwise, \( \Delta y \sim \mathcal{N}(0, H\Sigma H/t) \).

We first write:
\begin{equation}
    \left|\prod_{j \in T} \left(y_j^2 + 2yy_j - \frac{(H\Sigma H)_j}{t}\right)\right| = \sum_{T' \subset T} \prod_{j \in T/T'} \left(-\frac{(H\Sigma H)_j}{t}\right) \sum_{T_2 \subset T'} \prod_{j \in T_2} y_j^2 \prod_{j \in T'/T_2} 2yy_j.
\end{equation}
Suppose we want to compute 
\begin{align*}
    & \left| \Exp_{(x,y) \in \mathfrak{G}}[\prod_{i \in S}x_i\prod_{j \in T}(y_j^2+2yy_j-(H\Sigma H)_j/t)]-    \Exp_{(x,y) \in \mathfrak{U}}[\prod_{i \in S}x_i\prod_{j \in T}(y_j^2+2yy_j-(H\Sigma H)_j/t)] \right| \nonumber \\ 
    & \leq \sum_{T' \subset T}\prod_{j \in T/T'}\left|(\frac{H\Sigma H_j}{t})\right|\sum_{T_2 \subset T'}|2y|^{|T'|-{T_2}} \left( \Exp_{(x,y) \in \mathfrak{G}}[\prod_{i \in S}x_i\prod_{j \in T_2}y_j^2 \prod_{j \in T'/T_2}y_j]-\Exp_{(x,y) \in \mathfrak{U}}[\prod_{i \in S}x_i\prod_{j \in T_2}y_j^2 \prod_{j \in T'/T_2}y_j] \right) \nonumber \\ 
    & \leq \sum_{T' \subset T}(\frac{\epsilon_0}{t})^{|T|-|T'|}\sum_{T_2 \subset T'} \left( \Exp_{(x,y) \in \mathfrak{G}}[\prod_{i \in S}x_i\prod_{j \in T_2}y_j^2 \prod_{j \in T'/T_2}y_j]-\Exp_{(x,y) \in \mathfrak{U}}[\prod_{i \in S}x_i\prod_{j \in T_2}y_j^2 \prod_{j \in T'/T_2}y_j] \right)  \nonumber \\ 
    &\leq 2^{T}\sum_{|T'|=0}^{T}(\frac{\epsilon_0}{t})^{|T|-|T'|}\sum_{T_2 \subset T'}2(|S|+|T'|+|T_2|)!\frac{(\epsilon_{0}/t)^{|S|/2+|T'|/2+|T_2|/2}}{N^{m}}  
    \nonumber \\
         & \leq  2^{T}\sum_{|T'|=0}^{T}(\frac{\epsilon_0}{t})^{|T|-|T'|}2^{T'} \frac{4(|S|+|T'|)!(\epsilon_{0}/t)^{|S|/2+|T'|/2}}{N^{m}} \leq 
    \nonumber (\epsilon_{0}/t)^{|S|/2+|T|}\frac{8 \cdot 2^{|T|}|S|!}{N^{m}}.\\
\end{align*} 
From (2) to (3), we use $\left|(\frac{H\Sigma H_j}{t})\right| \leq \frac{\epsilon_0}{t}, |2y| \leq 1.$ (3) to (4), we use $\sum_{T' \subset T} \leq 2^{T} \sum_{T'=0}^{T}$.
In (4) to (5), we use $\sum_{T_2 \subset T'}\frac{(\epsilon_{0}/t)^{|S|/2+|T'|/2+|T_2|/2}}{N^{m}}2(|S|+|T'|+|T_2|)! \leq 2^{T'} 4(|S|+|T'|)!\frac{(\epsilon_{0}/t)^{|S|/2+|T'|/2}}{N^{m}}$ if we pick $t$ large enough. We use a similar reason for $\sum_{T'=0}.$

\end{proof}

Finally, we can bound $(*)$:
\begin{align*}
        &(*)=\left| \Exp_{z=(x,y) \in \mathfrak{G},\rho}[F_{\rho}(2 \Delta x, 2 \Delta y')]-    \Exp_{z=(x,y) \in \mathfrak{U},\rho}[F_{\rho}(2 \Delta x, 2 \Delta y')] \right|  \nonumber \\ &\leq \sum_{V \subset [2N], V=S \cup T}|\hat{F}(V)|2^N(\epsilon_{0}/t)^{|S|/2+|T|}\frac{8 \cdot 2^{|T|}|S|!}{N^{m}} \nonumber \leq O(\frac{\epsilon_0 \cdot l}{t N^{m}}).\\
\end{align*}
Here we can see the dominate term comes from $|V|=2$ and we use $L_{1,2}(F_\rho) \leq l$.

Finally, we can finish the proof of \thm{mainUD} 

\begin{theorem}\label{thm:D7}
Let \( F : \{\pm 1\}^{2N} \to \{\pm 1\} \) be a multilinear function with bounded spectral norm \( L_{1,2}(F) \leq l \), then
\begin{equation}
    \left| \Exp_{(x,y) \sim \mathscr{G}}[F(x,y')]- \Exp_{(x,y) \sim \mathscr{U}}[F(x,y')]\right| \leq O\left(\frac{\epsilon_0 \cdot l}{N^m}\right).
\end{equation}
\end{theorem}

Let \( t = N^c \) for a sufficiently large constant \( c \), \( p = \frac{1}{\sqrt{t}} \), and define \( z^{i,j} = p\left(\sum_{k=1}^{i} z^k + \sum_{k=i+1}^{i+j} z^k\right) \). For all \( k \leq i \), \( z^{\{k\}} \sim \mathfrak{G} \) and for all \( k > i \), \( z^{\{k\}} \sim \mathfrak{U} \). This forms a mixed path. It is trivial to see that \( z^{t,0} \sim \mathfrak{G} \) and \( z^{0,t} \sim \mathfrak{U} \).

We denote \( z^{i,j} = (x^{i,j}, y^{i,j}) \) and \( y'^{i,j} = (y^{i,j})^2 - \frac{(i+j)(H\Sigma H) }{t} \).

Since the distribution of \( z^{t,0} \) is a multivariate Gaussian distribution with the same expectation and covariance matrix as \( \mathfrak{G} \) and is similar to \( \mathfrak{U} \), it will be sufficient to bound:
\begin{equation}
\left| \Exp[F(\text{trnc}(x^{t,0}, y'^{t,0}))]-\Exp[F(\text{trnc}(x^{0,t}, y'^{0,t}))] \right|.
\end{equation}

We can express this as a difference between mixed paths:
\begin{align*}
   & \left| \Exp_{z \sim \mathscr{G}}[F(z)] - \Exp_{z \sim \mathscr{U}}[F(z)] \right| = \left| \Exp_{z \sim \mathfrak{G}}[\text{trnc}(F)(z)] - \Exp_{z \sim \mathfrak{U}}[\text{trnc}(F)(z)] \right| \\
   & \leq \sum_{i=0}^{t-1} \left| \Exp[F(\text{trnc}(z^{t-i,i}))] - \Exp[F(\text{trnc}(z^{t-i-1,i+1}))] \right|.
\end{align*}

Let $E$ be the event that $z^{i,j} \in [-1/2, 1/2]^{2N}$. It is also clear that if $z^{i,j} \in [-1/2, 1/2]^{2N}$ then $z'^{i,j} \in [-1/2, 1/2]^{2N}$. Each $l$ entry in $z^{i,j}$ is distributed according to $\mathcal{N}(0,p^2 i \epsilon)$ for some $\epsilon \leq \epsilon_0$, and we have:
\begin{equation}
    \Pr[|z^{i, j}_l| \geq \frac{1}{2}] \leq \Pr[|\mathcal{N}(0, \epsilon_0)| \geq \frac{1}{2}] \leq e^{-\frac{1}{8 \epsilon_0}} \leq \frac{1}{N^{c+1}}.
\end{equation}
By the union bound, we have $\Pr[E] \geq 1 - 2N^{-c}$. By \thm{f7}, used with $z_0 = z^{i,j}$, we have that conditioned on the event $E$:
\begin{equation}
    |\Exp[F(x^{i+1,j}, y'^{i+1,j})|E] - \Exp[F(x^{i,j+1}, y'^{i,j+1})|E]| \leq O\left(\frac{\epsilon_0 \cdot l}{t N^m}\right).
\end{equation}
Using \lem{r2}, we have:
\begin{align*}
    |\Exp[F(\text{trnc}(x^{i+1,j}, y'^{i+1,j})) - \Exp[F(x^{i+1,j}, y'^{i+1,j})|E]| & \leq O(N^{-c}), \\
    |\Exp[F(\text{trnc}(x^{i,j+1}, y'^{i,j+1})) - \Exp[F(x^{i,j+1}, y'^{i,j+1})|E]| & \leq O(N^{-c}).
\end{align*}
Now we have:
\begin{align*}
    &|\Exp[F(\text{trnc}(z^{i+1, j}))|E] - \Exp[F(\text{trnc}(x^{i,j+1}, y'^{i,j+1}))|E]| \\
    &\quad \leq |\Exp[F(\text{trnc}(x^{i+1,j}, y'^{i+1,j}))] - \Exp[F(x^{i+1,j}, y'^{i+1,j})|E]| + \\
    &\quad\quad |\Exp[F(\text{trnc}(x^{i,j+1}, y'^{i,j+1}))] - \Exp[F(x^{i,j+1}, y'^{i,j+1})|E]| \\
    &\quad\quad\quad + |\Exp[F(x^{i+1,j}, y'^{i+1,j})] - \Exp[F(x^{i+1,j}, y'^{i+1,j})|E]| \\
    &\quad\quad\quad + |\Exp[F(x^{i,j+1}, y'^{i,j+1})] - \Exp[F(x^{i,j+1}, y'^{i,j+1})|E]| \\
    &\quad \leq O\left(\frac{\epsilon_0 \cdot l}{t N^m}\right) + O(N^{-c}).
\end{align*}
When $E$ does not hold, the difference between $F(\text{trnc}(x^{i+1,j}, y'^{i+1,j}))$ and $F(\text{trnc}(x^{i,j+1}, y'^{i,j+1}))$ is at most 2 since $F$ maps $[-1, 1]^{2N} \to [-1, 1]$.
Thus:
\begin{align*}
   |\Exp[F(\text{trnc}(x^{i+1,j}, y'^{i+1,j}))] - \Exp[F(\text{trnc}(z^{i, j+1}))]| & \leq \\
   |\Exp[F(\text{trnc}(x^{i+1,j}, y'^{i+1,j}))|E] - \Exp[F(\text{trnc}(z^{i, j+1}))|E]| + 2\Pr[\neg E] & \leq \\
   O\left(\frac{\epsilon_0 \cdot l}{t N^m}\right) + O(N^{-c}).
\end{align*}

\noindent By setting $\epsilon_0=1/(c_1n)$, we can have the following corollary. 
\begin{corollary}[{\cite{bassirian2021certified}}]
    If $\Sigma= \Id/(c_1n)$ for large enough $c_1.$ Then $|\Exp_{z \sim \mathscr{G}}F(x,y')-\Exp_{z \sim \mathscr{U}}F(x,y')|$ is at most $\polylog(N)/\sqrt{N}.$  
\end{corollary}

Hence, we can finish the proof \thm{mainUD}. We think this proof can also be simplified by stochastic calculus using \cite{wu2020stochastic} and also a rather useful tool called Gaussian interpolation \cite{bansal2021k}.

\bibliographystyle{alpha}
\bibliography{sample}

\appendix

\section{Useful Theorems}\label{app1}

\begin{theorem}[Hoeffding's inequality]\cite{krafft1969note}\label{thm:hoe}
Let $X_1, ..., X_n$ be independent random variables such that $0 \leq X_{i} \leq 1$ almost surely. Consider the sum of these random variables, $S_{n}=X_{1}+\cdots +X_{n}.$
Then we get the inequality
\begin{equation}
{\displaystyle {\begin{aligned}\operatorname {P} \left(S_{n}-\Exp \left[S_{n}\right]\geq t\right)&\leq \exp \left(-{\frac {2t^{2}}{n}}\right).\\
\operatorname {P} \left(\Exp \left[S_{n}\right]-S_n \geq t\right)&\leq \exp \left(-{\frac {2t^{2}}{n}}\right). \\
\operatorname {P} \left(\left|S_{n}-\Exp \left[S_{n}\right]\right|\geq t\right)&\leq 2\exp \left(-{\frac {2t^{2}}{n}}\right).\end{aligned}}}
\end{equation}
for all $t\geq 0.$
\end{theorem}

Isserlis' theorem for the expectations of products of jointly Gaussian random variables:
\begin{theorem}[Isserlis' theorem]\cite{michalowicz2011general}\label{thm:is}
If 
${\displaystyle (X_{1},\dots ,X_{n})}$ is a zero-mean multivariate normal random vector, then
\begin{equation}
    {\displaystyle \Exp [\,X_{1}X_{2}\cdots X_{n}\,]=\sum _{p\in P_{n}^{2}}\prod _{\{i,j\}\in p}\Exp [\,X_{i}X_{j}\,]=\sum _{p\in P_{n}^{2}}\prod _{\{i,j\}\in p}\operatorname {Cov} (\,X_{i},X_{j}\,),}
\end{equation}
    where the sum is over all the pairings of $\{1,\ldots,n\}$  i.e. all distinct ways of partitioning $\{1,\ldots,n\}$
 into pairs 
${\displaystyle \{i,j\}}$, and the product is over the pairs contained in $p.$
If $n=2m$, there are $\frac{(2m)!}{m!2^m}$ terms in 
 the $\Sigma \prod$ expression for 
variables.
Moreover If 
${\displaystyle n=2m+1}$ is odd, we have ${\displaystyle \Exp [\,X_{1}X_{2}\cdots X_{2m+1}\,]=0.}$
\end{theorem}

\section{Noisy Quantum Circuit Model}\label{app:nisq}
In this section, we review our definition of a noisy quantum circuit model. Most of the definition is similar to \cite{chen2022complexity}. We begin by recalling the single-qubit depolarizing channel, \(D_\lambda\).
\begin{definition}[Single-qubit depolarizing channel]
Given $\lambda \in [0,1]$, we define the \emph{single-qubit depolarizing channel} to be $D_\lambda[\rho] \triangleq (1-\lambda) \rho + \lambda (I / 2)$, where $\rho$ is a single-qubit density matrix.
\end{definition}

\begin{definition}[Depth-$1$ unitary]
Given $n > 0$, an $n$-qubit unitary $U$ is a \emph{depth-$1$ unitary} if $U$ can be written as a tensor product of two-qubit unitaries.
\end{definition}

\noindent 
We define noisy quantum circuits with a noise level $\lambda$ and their output distributions as follows.

\begin{definition}[Output of a noisy quantum circuit]
Let $\lambda \in [0, 1]$ and $n \in \mathbb{N}$. 
    Given $T \in\mathbb{N}$ and a sequence of $T$ depth-$1$ unitaries $U_1, \ldots, U_T$, the output of the corresponding \emph{$\lambda$-noisy depth-$T$ quantum circuit} is a random $n$-bit string $s \in \{0, 1\}^n$ sampled from the distribution
    \begin{equation} \label{eq:prob-NQC}
        p(s) = \bra{s} D_\lambda^{\otimes n}\bigl[U_T \ldots D_\lambda^{\otimes n}\bigl[U_2 D_\lambda^{\otimes n}\bigl[U_1 D_\lambda^{\otimes n}[\ketbra{0^n}{0^n}] U_1^\dagger\bigr] U_2^\dagger \bigr] \ldots U_T^\dagger \bigr] \ket{s} \, ,
    \end{equation}
where every quantum operation is followed by a layer of the single-qubit depolarizing channel. When $\lambda = 0$, we say that this circuit is \emph{noiseless}.
\end{definition}

\begin{definition}[Noisy quantum circuit oracle]
    We define $\mathrm{NQC}_{\lambda}$ to be an oracle that takes in an integer $n$ and a sequence of depth-$1$ $n$-qubit unitary $\{U_k\}_{k=1,\ldots,T}$ for any $T\in\mathbb{N}$ and outputs a random $n$-bit string $s$ according to Eq.~\eqref{eq:prob-NQC}.
    We define the time to query $\mathrm{NQC}_{\lambda}$ with $T$ depth-$1$ $n$-qubit unitaries to be $\Theta(nT)$, which is linear in the time to write down the input to the query.
\end{definition}

\noindent We now define $\BQP_{\lambda}$ algorithms, which are classical algorithms with access to the noisy quantum circuit oracle and $\lambda: \mathbb{N} \to (0,1)$. This provides a formal definition for hybrid noisy quantum-classical computation.

\begin{definition}[$\BQP_{\lambda}$ algorithm]
    A \emph{$\BQP_{\lambda}$ algorithm with access to $\lambda$-noisy quantum circuits} is defined as a probabilistic Turing machine $M$ that can query $\mathrm{NQC}_{\lambda}$ to obtain an output bitstring $s$ for any number of times and is denoted as $A_\lambda \triangleq M^{\mathrm{NQC}_\lambda}$. The runtime of $A_\lambda$ is given by the classical runtime of $M$ plus the sum of the times to query $\mathrm{NQC}_\lambda$.
\end{definition}

\noindent The $\BQP_{\lambda}$ complexity class for decision problems is defined as follows.

\begin{definition}[$\BQP_{\lambda}$ complexity]
    A language $L \subseteq \{0, 1\}^*$ is in $\BQP_{\lambda}$  if there exists a $\BQP_{\lambda}$  algorithm $A_\lambda$ such that decides $L$ in polynomial time, that is, such that
    \begin{itemize}
        \item for all $x \in \{0, 1\}^*$, $A_\lambda$ produces an output in time $\mathrm{poly}(|x|)$, where $|x|$ is the length of $x$;
        \item for all $x \in L$, $A_\lambda$ outputs $1$ with probability at least $2/3$;
        \item for all $x \not\in L$, $A_\lambda$ outputs $0$ with probability at least $2/3$.
    \end{itemize}
\end{definition}

\subsection{Algorithms with oracle access}
\label{subsec:oracle}

\begin{definition}[Classical oracle $\O$]
    A \emph{classical oracle} $\O$ is a function from $\{0, 1\}^n$ to $\{0, 1\}^m$ for some $n, m \in \mathbb{N}$.
    The $(n+m)$-qubit unitary $U_\O$ corresponding to the classical oracle $O$ is given by $U_{\O} \ket{x}\ket{y} = \ket{x} \ket{y \oplus \O(x)}$ for all $x \in \{0, 1\}^n, y \in \{0, 1\}^m$.
\end{definition}

\begin{definition}[Classical algorithm with access to $\O$]
    A \emph{classical algorithm $M^{\O}$ with access to $\O$} is a probabilistic Turing machine $M$ that can query $O$ by choosing an $n$-bit input $x$ and obtaining the $m$-bit output $\O(x)$.
\end{definition}    

\begin{definition}[Quantum algorithm with access to $\O$]
    A \emph{quantum algorithm $Q^{\O}$ with access to $\O$} is a uniform family of quantum circuits $\{U_n\}_n$, where $U_n$ is an $n'$-qubit quantum circuit given by
    \begin{equation}
        U_n \triangleq V_{n, k} (U_{\O} \otimes \Id) \cdots (U_{\O} \otimes \Id) V_{n, 2} (U_{\O} \otimes \Id) V_{n, 1} \, ,
    \end{equation}
    for some integer $k \in \mathbb{N}$ and $n'$-qubit unitaries $V_{n, 1}, \ldots, V_{n, k}$ given as the product of many depth-$1$ unitaries. Here, $\Id$ denotes the identity matrix over $n' - n$ qubits.
\end{definition}

\noindent For fixed function: $\lambda: \mathbb{N} \to \mathbb{R}.$ We now present the definition of $\BQP_{\lambda}$ algorithms with access to the classical oracle ${\O}$, which requires first defining noisy quantum circuit oracles with access to ${\O}$.

\begin{figure}[h]
\centering
\includegraphics[scale=0.3]{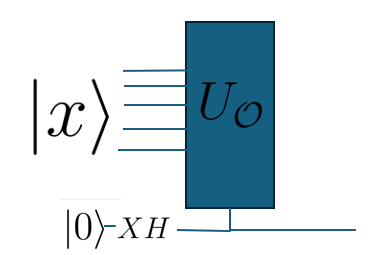}
\caption{From Classical Oracle to Phase Oracle
, We can observe that this control qubit with probability $p_1$, will remain in the $\ket{-}$ state; otherwise, it becomes $\ket{+}$ where $p_1 = (1-\lambda)(1-\lambda/2) + \lambda/2$.}
\label{Fig:Oracle}
\end{figure}

\begin{definition}[Noisy quantum circuit oracle with access to $\O$]
    We define $\mathrm{NQC}^\O_{\lambda}$ to be an oracle that takes in an integer $n'$ and a sequence of $n'$-qubit unitaries $\{U_k\}_{k=1,\ldots,T}$ for any $T\in\N$, where $U_k$ can either be a depth-$1$ unitary or $U_\O \otimes I$, to a random $n$-bit string $s$ sampled according to the distribution
    \begin{equation} \label{eq:prob-NQC-O}
        p(s) = \bra{s} D_\lambda^{\otimes n'}\bigl[U_T \ldots D_\lambda^{\otimes n'}\bigl[U_2 D_\lambda^{\otimes n'}\bigl[U_1 D_\lambda^{\otimes n'}[\ketbra{0^{n'}}{0^{n'}}] U_1^\dagger\bigr] U_2^\dagger \bigr] \ldots U_T^\dagger \bigr] \ket{s} \, .
    \end{equation}
\end{definition}

\begin{definition}[$\BQP_\lambda$ algorithm with access to $\O$]
    Let $\lambda: \mathbb{N} \to \mathbb{R}.$. A $\BQP_\lambda$ algorithm $A_\lambda^\O = (M^{\mathrm{NQC}_{\lambda}})^\O$ with access to $\O$ is a probabilistic Turing machine $M$ that has the ability to classically query $\O$ by choosing the $n$-bit input $x$ to obtain the $m$-bit output $\O(x)$, as well as the ability to query $\mathrm{NQC}^O_{\lambda}$ by choosing $n'$ and $\{U_k\}_{k=1,\ldots,T}$ to obtain a random $n'$-bit string $s$. The runtime of $A^\O_\lambda$ is given by the sum of the classical runtime of $M$, the number of classical queries to $\O$, and the sum of the times to query $\mathrm{NQC}^\O_\lambda$.
\end{definition}

\noindent With this definition in hand, we can extend the usual notions of relativized complexity to $\BQP_{\lambda}$:

\begin{definition}[Relativized $\BQP_{\lambda}$]
    Given a sequence of oracles $\O: \{0,1\}^n \to \{0,1\}^{m(n)}$ parametrized by $n\in\mathbb{N}$, a language $L\subseteq\{0,1\}^*$ is in $\BQP_{\lambda}^\O$ if $\BQP_{\lambda}$ algorithm $A^\O_\lambda$ with access to $\O$ that decides $L$ in polynomial time.
\end{definition}

\textbf{Remark on Forrelation:}\label{rm-f}
We describe a single oracle $\mathcal{O}$ from $\{0, 1\}^n$ to $\{0, 1\}$. It is trivial to see that $U_{\mathcal{O}} \ket{x}\ket{+} = \ket{x} \ket{+}$, and $U_{\mathcal{O}} \ket{x}\ket{-} = (-1)^{\mathcal{O}(x)}\ket{x} \ket{-}$. When we consider the noisy circuit in Fig.~\ref{Fig:Oracle}, using a similar calculation as in \thm{rigor-4}, we can observe that this control qubit, with probability $p_1$, will remain in the $\ket{-}$ state; otherwise, it transitions to $\ket{+}$. Here, $p_1 = (1-\lambda)(1-\lambda/2) + \lambda/2$. Thus, with probability $p_1$, the oracle behaves correctly; otherwise, we can consider the oracle $\mathcal{O}_f$ as an identity.
In the Forrelation problem, if the oracle malfunctions, the accept and reject probabilities will be the same whether $f, g$ are Forrelated or uniformly distributed, since now $f$ can be viewed as identity. Otherwise, we have the correct Fourier sampling circuit. Hence, the advantage still remains but only shrinks by $(1-p_1)$, which is a single constant.

\end{document}